\documentclass[12pt]{article}

\usepackage[top=1in, left=1in, right=1in,
bottom=1in]{geometry}

\usepackage{latexsym}
\usepackage{amssymb,amsmath, bm, bbm}
\usepackage{graphicx}
\usepackage{marvosym}
\usepackage{lscape}
\usepackage{rotating}
\usepackage{setspace}
\usepackage{threeparttable}
\usepackage{booktabs}
\usepackage[compact]{titlesec}

\usepackage{caption}
\usepackage{subcaption}
\usepackage[shortlabels]{enumitem}

\titleformat{\section}
  {\normalfont\fontsize{11}{15}\bfseries}{\thesection}{1em}{}

\titleformat{\subsection}
  {\normalfont\fontsize{11}{15}\bfseries}{\thesubsection}{1em}{}

\usepackage{natbib}
\usepackage{color}
\usepackage[pdftex, bookmarksopen=true, bookmarksnumbered=true,
pdfstartview=FitH, breaklinks=true, urlbordercolor={0 1 0}, citebordercolor={0 0 1}, colorlinks,allcolors=blue]{hyperref}

\usepackage{xr}
\usepackage{dcolumn}
\newcolumntype{.}{D{.}{.}{-1}}
\newcolumntype{d}[1]{D{.}{.}{#1}}

\usepackage{amsthm}
\theoremstyle{definition}
\newtheorem{theorem}{Theorem}

\newtheorem{assumption}{Assumption}

\newcommand{\R}{\ensuremath{\mathbb{R}}}

\newcommand{\bbone}{\ensuremath{\mathbbm{1}}}

\newcommand{\E}{\ensuremath{\mathbb{E}}}

\newcommand{\Cov}{\text{Cov}}
\newcommand{\Var}{\text{Var}}

\newcommand{\imbal}{\text{imbalance}}
\newcommand{\imbalance}{\text{imbalance}}

\def\spacingset#1{\renewcommand{\baselinestretch}%
{#1}\small\normalsize} \spacingset{1}

\begin{document}

\pagestyle{plain}

\newcommand{\blind}{0}

\newcommand{\tit}{\Large Approximate Balancing Weights for Clustered Observational Study Designs}

\if0\blind

{\title{\tit\thanks{This research is supported by the Institute of Education Sciences, U.S. Department of Education, through Grant R305D210014. The opinions expressed are those of the authors and do not represent views of the Institute or the U.S. Department of Education. One of the datasets used for this study was purchased with a grant from the Society of American Gastrointestinal and Endoscopic Surgeons. Although the AMA Physician Masterfile data is the source of the raw physician data, the tables and tabulations were prepared by the authors and do not reflect the work of the AMA. The Pennsylvania Health Cost Containment Council (PHC4) is an independent state agency responsible for addressing the problems of escalating health costs, ensuring the quality of health care, and increasing access to health care for all citizens. While PHC4 has provided data for this study, PHC4 specifically disclaims responsibility for any analyses, interpretations or conclusions. Some of the data used to produce this publication was purchased from or provided by the New York State Department of Health (NYSDOH) Statewide Planning and Research Cooperative System (SPARCS). However, the conclusions derived, and views expressed herein are those of the author(s) and do not reflect the conclusions or views of NYSDOH. NYSDOH, its employees, officers, and agents make no representation, warranty or guarantee as to the accuracy, completeness, currency, or suitability of the information provided here. This publication was derived, in part, from a limited data set supplied by the Florida Agency for Health Care Administration (AHCA) which specifically disclaims responsibility for any analysis, interpretations, or conclusions that may be created as a result of the limited data set.}}

\author{Eli Ben-Michael\thanks{Carnegie Mellon University, Pittsburgh, PA, Email: ebenmichael@cmu.edu}
\and Lindsay Page\thanks{Brown University, Providence, RI, Email: lindsay\_page@brown.edu}
\and Luke Keele\thanks{University of Pennsylvania, Philadelphia, PA, Email: luke.keele@gmail.com}
}

\date{\today}

\maketitle
}\fi

\if1\blind
\title{\tit}
\maketitle
\fi

\begin{abstract}

In a clustered observational study, a treatment is assigned to groups and all units within the group are exposed to the treatment. We develop a new method for statistical adjustment in clustered observational studies using approximate balancing weights, a generalization of inverse propensity score weights that solve a convex optimization problem to find a set of weights that directly minimize a measure of covariate imbalance, subject to an additional penalty on the variance of the weights. We tailor the approximate balancing weights optimization problem to both adjustment sets by deriving an upper bound on the mean square error for each case and finding weights that minimize this upper bound, linking the level of covariate balance to a bound on the bias. We implement the procedure by specializing the bound to a random cluster-level effects model, leading to a variance penalty that incorporates the signal signal-to-noise ratio and penalizes the weight on individuals and the total weight on groups differently according to the the intra-class correlation.

\end{abstract}

\begin{center}
\end{center}

\clearpage
\doublespacing

\spacingset{1.5}

\section{Introduction}

In a study of comparative effectiveness, researchers seek to understand whether a treatment has a causal effect on an outcome of interest for a set of study units. In the causal inference literature, treatment assignment is a critical element of the study design \citep{Rubin:2007,Rubin:2008}. Two key components of treatment assignment are whether the intervention is randomized or not and whether it is grouped or not. First, when interventions are randomly assigned, differences between treated and control groups can be interpreted as causal effects. In contrast, when subjects select their own treatments, stronger assumptions are needed to identify causal effects. Second, interventions may be assigned to individual units or to intact groups. For example, given that students are grouped in schools, a treatment may be assigned to all students in some schools and withheld from all students in other schools. When treatment is randomly assigned and grouped, the design is often called a clustered randomized trial (CRT) \citep{raudenbush1997statistical, hedges2007intraclass}. In many cases, however, treatments are grouped but non-randomly assigned. This design is referred to as the clustered observational study (COS) \citep{Keele:2016b,pagedesign2019}. COS designs differ from non-clustered observational studies, both because they rely on different identification assumptions and because they require different methods for analysis \citep{ye2022clustered}.

In a COS design, as is true for any observational study, differences between treated and control outcomes may reflect initial differences in the treated and control groups rather than treatment effects \citep{Cochran:1965,Rubin:1974}. As such, analysis for COS designs requires statistical adjustment methods to account for observed confounders. Adjustment methods for COS designs, however, may need to remove treated and control differences in the distributions of covariates at the cluster level, the unit level, or perhaps both levels. Recent work has also shown that statistical adjustment in the COS design needs to reflect key substantive knowledge of how the clustered treatment is assigned and whether differential selection of clusters is present \citep{ye2022clustered}. In short, analysts require a statistical adjustment strategy that takes into account key aspects of the COS design.

In this article, we develop an approximate balancing weights estimator tailored to the COS context. Balancing weight estimators solve a convex optimization problem to find a set of weights that directly minimize a measure of covariate imbalance subject to an additional constraint or penalty on the complexity of the weights \citep{Hainmueller2011,zubizarreta2015stable,ben2021balancing}. Approximate balancing weights are a generalization of the standard inverse propensity score estimator.  Approximate balancing weight methodology, however, requires a number of key innovations to be used in the COS context. First, we write separate objective functions with different balance measures and variance penalties for the two possible estimands in a COS. The first estimand adjusts for both cluster- and unit-level covariates while the second adjusts for cluster-level covariates only. For both estimands, we derive a bound on the mean square error of a general weighting estimator, and find weights that minimize this upper bound, showing that the level of balance between the treated and re-weighted control group gives a bound on the bias. Next, we write the variance penalty as a cluster-level random effects model, and we show that the variance penalty is comprised of two parts. The first component is the signal-to-noise ratio which measures the overall impact of the variance relative to bias. The second component is the intra-class correlation. These components correspond to hyper-parameters in the corresponding balancing optimization problem, and we develop a data-driven approach for selecting them. For cases where it is sufficient to condition on cluster-level covariates only, we show that it can be more efficient to further adjust for unit-level covariates when they are strongly predictive of the outcome or if the intra-class correlation is moderate to high. In applications with poor overlap or many covariates, it can be difficult to find weights that achieve good balance. For these cases, we develop two separate extensions: using an outcome model to perform additional bias correction, and adjusting the balancing optimization problem to find the maximally overlapping set between the treated and control group. Finally, we derive methods for variance estimation and show how to conduct inference based on asymptotic normality. In a series of simulations, we find that balancing weights are superior to multilevel matching --- which is also tailored to the COS design --- in terms of bias and variance reduction. We provide further comparisons with two empirical applications estimating the effects of Catholic schools and surgical training. Overall, we find that balancing weights produce superior balance relative to extant matching methods and have larger effective sample sizes.

Our article proceeds as follows. In Section~\ref{sec:cos}, we review the details of the COS design. In Section~\ref{sec:bwt} we develop approximate balancing weights for the COS design and discuss extensions and inference in Section~\ref{sec:ext}. In Section~\ref{sec:sim}, we evaluate our methods in a simulation study. In Section~\ref{sec:apps}, we analyze the data from two COS designs: one from education and one from health services research. In Section~\ref{sec:dis}, we conclude and discuss directions for future work.

\section{The COS Design}
\label{sec:cos}

First, we review the formal aspects of the COS design and how the process of differential selection leads to different estimands and adjustment strategies. See \citet{ye2022clustered} for a detailed treatment of identification issues in the COS design.

\subsection{Notation}

We consider a setup with $m$ clusters, with $n_\ell$ units in cluster $\ell$, and $n = \sum_{\ell = 1}^m n_\ell$ total units. We denote the treatment status of cluster $\ell$ as $A_\ell \in \{0,1\}$, $n_1 \equiv \sum_{\ell = 1}^m (1 - A_\ell)n_\ell$ as the total number of treated units, and $n_0 \equiv n - n_1$ as the total number of control units. Each collection of treatment status vectors $\bm{a} = (a_1,\ldots,a_m) \in \{0,1\}^m$ is associated with a vector of potential \emph{cluster assignments}, $\bm{J(a)} = (J_1(\bm{a}),\ldots, J_n(\bm{a})) \in \{1,\ldots,m\}^n$, where $J_i(\bm{a}) \in \{1,\ldots,m\}$ denotes the cluster that unit $i$ would belong to under overall treatment allocation $\bm{a}$. We denote the observed cluster assignments as $\bm{J} = \bm{J}(\bm{A})$. Each unit $i$ has a \emph{potential outcome} $Y_i(a_{J_i(\bm{a})}, \bm{J}(\bm{a}))$ corresponding to the outcome that would be observed if its associated cluster has treatment status $a_{J_i}$ and the overall cluster assignment is $\bm{J}(\bm{A})$. Note that here we have followed \citet{ye2022clustered} and assumed that a unit's potential outcome only depends on its own cluster's treatment status and not the treatment status of other clusters, except through the potential cluster assignments. We further assume that potential outcomes are independent across clusters --- i.e., they are independent for unit $i$ and $i'$ if $J_i(\bm{a}) \neq J_{i'}(\bm{a})$ --- but may be dependent within clusters. We denote the observed outcomes as $Y_i(A_{J_i}, \bm{J})$. We also assume that we observe unit-level covariates $X_i \in \mathcal{X}$ and that the cluster assignments lead to potential cluster-level covariates, $W_\ell(\bm{J}(\bm(a))) \in \mathcal{W}$, which may include summaries of the unit-level covariates and so can depend on the potential cluster assignments. We let $W_\ell = W_\ell(\bm{J}(\bm(A)))$ denote the observed cluster-level covariates, and $\bm{W}$ denote the collection of observed cluster-level covariates for all clusters. Taken together, our observed data consists of tuples $(X_i, J_i, W_{J_i}, A_{J_i}, Y_i)$.

Next, we define several conditional expectations of the observed outcomes that we will use throughout. First, we denote $m_w(a, w) \equiv \E[Y_i \mid A_{J_i} = a, W_{J_i} = w]$ as the expected outcome for unit $i$, conditioned on its cluster having treatment status $a$ and covariates $w$. Next, we use $m_{wx}(a, w, x) \equiv \E[Y_i \mid A_{J_i} = a, W_{J_i} = w, X_i = x]$ to denote the expected outcome if we add more information and additionally condition on unit-level covariates $x$. Next, let $e(w) = P(A_\ell = 1 \mid W_\ell = w)$ denote the propensity score, the probability that a cluster is treated, conditioning on cluster-level covariates, and $e(w, x) = P(A_\ell = 1 \mid W_\ell = w, X_i = x)$ denote the propensity score conditioning on cluster and unit-level covariates.

\subsection{Estimand, Designs, and Assumptions}

\citet{ye2022clustered} delineate two primary COS designs. The first we denote as the Cluster-Unit Design (CUD), and the second we denote as the Cluster-Only Design (COD). Both designs focus on a common target causal estimand which is the average treatment effect for the treated units under the observed cluster assignments $\bm{J}$,
\begin{equation}
  \label{eq:att}
  \tau \equiv \underbrace{\E\left[Y_i(1, \bm{J}) \mid A_{J_i} = 1\right]}_{\mu_1} - \underbrace{\E\left[Y_i(0, \bm{J}) \mid A_{J_i} = 1\right]}_{\mu_0}.
\end{equation}
This estimand measures the effect of treatment for units in treated clusters, keeping the cluster assignments fixed. The first term, $\mu_1$ can be written as the expectation of the observed outcomes among units in the treated clusters, $\mu_1 = \E[Y_i \mid A_{J_i} = 1]$.
However, identifying and estimating $\mu_0$, the mean counterfactual outcome if those clusters had in fact been assigned to control, is more difficult. Next, we review the key assumptions required for identification of this target causal estimand under the two designs.

\subsubsection{Cluster-Unit Design}

The key distinction between the identification strategies for these two designs is whether the units within the clusters respond to treatment assignment to clusters. That is, under the CUD, the mix of units within clusters is changed by the fact that clusters are treated. \citet{ye2022clustered} refer to such bias as differential selection. Differential selection could occur if the student mix in a school changes in response to the school being treated. For example, parents may opt to enroll their child in a school implementing a whole-school curricular reform. For the CUD, we need to account for unit selection into the clusters, which can depend on the treatment assignments. Identification of $\mu_0$ under the CUD requires the following set of assumptions:
  \begin{itemize}
    \item For two cluster assignment vectors $\bm{j}$ and $\bm{j}'$ such that the treatment assignments are the same, $a_{j_i} = a_{j_i}'$, and the cluster-level covariates are unchanged $W_{j_i} = W_{j'_i}$ for unit $i$, the potential outcomes are equal $Y_i(a_{j_i}, \bm{j}) = Y_i(a_{j'_i}, \bm{j'})$.
    \item Denote $Y_i(a, w) = Y_i(a_{j_i} = a, W_{j_i} = w)$. For every unit $i$, cluster $\ell$, treatment value $a$, and cluster-level covariates value $w$, $A_{J_i} \perp Y_i(a, w) \mid W_{J_i}, X_i$.
    \item  $e(w,x) > 0$ for all $w,x \in \mathcal{W,X}$.
    \end{itemize}
Taken together, these assumptions form an ignorability assumption that includes both cluster-level covariates and unit-level covariates that drive treatment selection. Here, we identify $\mu_0$ as
\[
  \mu_0 = \E[m_w(0, W_{J_i}, X_i) \mid A_{J_i} = 1] \equiv \mu_{0wx}.
\]

\subsection{Cluster-Only Design}

For the COD, treatment assignment is restricted to cases where either treatment assignment occurs after the unit-cluster pairing, or units are blinded to the cluster assignments before pairing.  As such, treatment assignment only depends on cluster-level covariates. Identification of $\mu_0$ under the COD requires the following set of assumptions:

\begin{itemize}
  \item Cluster assignments are not affected by treatment: $J_\ell(\bm{a}) = J_\ell$ for all $\bm{a} \in \{0,1\}^m$.
  \item For every unit $i$, cluster $\ell$, cluster assignment vector $\bm{j}$, and treatment value $a$, $A_{J_i} \perp (\bm{J}, X_i, Y_i(a, \bm{j})) \mid W_{J_i}$.
  \item $e(w) > 0$ for all $w \in \mathcal{W}$.
  \end{itemize}
Here, conditioning on the cluster-level covariates is sufficient to remove confounding. As such, under the COD, the estimated for the expected outcome conditional on the cluster being assigned to the control condition and the cluster-level covariates is identified as:
\[
  \mu_0 = \E[m_w(0, W_{J_i}) \mid A_{J_i} = 1] = \mu_{0w}.
\]

The CUD requires a stronger assumption about the form of the potential outcomes but a weaker assumption on the treatment selection process than the COD. However, the key distinction between these two identification strategies is the conditioning sets. For the CUD, we must condition on both unit and cluster-level covariates. For the COD, we only need to condition on cluster-level covariates. As such, the COD is a special case of the CUD without unit-level covariates. Below, we consider the implications of conditioning on unit-level covariates when the assumptions of the COD hold, and we show that while unit-level covariates are unnecessary for identification, they may be beneficial in terms of increased precision. Next, we tailor approximate balancing weight methods to each design.

\section{Approximate Balancing Weights for the COS Design}
\label{sec:bwt}

Recent work has developed matching methods---known as multilevel matching---that are specifically tailored to the COS design \citep{Keele:2015,Keele:2016b,keele2021matching}. Here, we outline approximate balancing weights as an alternative to multilevel matching. Approximate balancing weights are a generalization of the standard inverse propensity score (IPW) estimator that solve  a convex optimization problem to find a set of weights that directly minimize a measure of covariate imbalance subject to an additional constraint or penalty on the complexity of the weights \citep{Hainmueller2011,zubizarreta2015stable}. Like matching, and unlike IPW, balancing weights are designed to directly target covariate balance in the estimation process, as opposed using the estimated probability of being selected for treatment. Like weighting, and unlike matching, balancing weights use more of the available data and so often have larger effective sample sizes.

Recent theoretical innovations have bolstered support for balancing weights. For example, it has been shown that balancing weights are implicitly estimates of the inverse propensity score, fit via a loss function that guarantees covariate balance \citep{Zhao2016a, Zhao2019,Wang2018,Chattopadhyay2020}. Researchers have proposed non-parametric extensions that allow for flexible specifications of the outcome model or propensity score \citep{hirshberg2019minimax,Hazlett2019}. Finally, recent literature on weighting allows these estimators to target different quantities such as sample overlap \citep{li2018balancing}. See \citet{ben2021balancing} for a general overview on balancing weights.

Approximate balancing weights for the COS design, however, require the development of specific objective functions that are tailored to the identification conditions outlined above. Next, we develop objective functions that target measures of covariate imbalance and variance penalties that are specific to COS designs. As we will see, one key challenge in developing objective functions for COS designs is accounting for how clustering affects the variance of the weights. Below, we consider a generic weighting estimator of the form
\[
  \hat{\mu}_0(\gamma) = \frac{1}{n_1}\sum_{\ell=1}^m (1 - A_\ell) \sum_{J_i = \ell} \gamma_i Y_i,
\]
where the weights $\gamma_i$ are independent of the outcomes.
We will first consider conditioning on cluster- and unit-level covariates to target $\mu_{0wx}$.
Then we will specialize this estimator to only condition on cluster-level covariates and target $\mu_{0w}$.
For both cases, we inspect the mean square error of the generic weighted average of control units' outcomes, then find weights that minimize an upper bound. Under both settings, we will consider the finite sample estimands:
\[
  \tilde{\mu}_{0wx} \equiv \frac{1}{n_1}\sum_{A_\ell = 1} \sum_{J_i = \ell} m_{wx}(0, W_{J_i}, X_i), \;\;\text{ and }\;\; \tilde{\mu}_{0w} \equiv \frac{1}{n_1} \sum_{A_\ell = 1} n_\ell m_w(0, W_\ell),
\]
which converge to our main estimands.

\subsection{Conditioning on cluster- and unit-level covariates}
\label{sec:bal_unit}
To target $\mu_{0wx}$, we use the following decomposition for the estimation error of the weighting estimator:
\begin{equation}
  \label{eq:bias_var_unit}
  \hat{\mu}_0\left(\gamma\right) - \tilde{\mu}_{0wx} =  \frac{1}{n_1}\sum_{A_\ell = 0} \sum_{J_i = \ell}\gamma_i m_{wx}(0, W_\ell, X_i) - \frac{1}{n_1} \sum_{A_\ell = 1} \sum_{J_i = \ell} m_{wx}(0, W_\ell, X_i) + \frac{1}{n_1}\sum_{A_\ell = 0}\sum_{J_i = \ell} \gamma_i \varepsilon_i.
\end{equation}
\noindent To understand this decomposition, we compute the design-conditional bias and variance. First, since the weights are independent of the outcomes, under the ignorability condition for the CUD, the conditional bias corresponds to the first term in the decomposition:
\begin{equation}
  \label{eq:bias_unit}
  \E\left[ \hat{\mu}_0\left(\gamma\right) - \tilde{\mu}_{0wx} \mid \bm{W}, \bm{A}, \bm{J}, \bm{X} \right] = \frac{1}{n_1}\sum_{A_\ell = 0} \sum_{J_i = \ell}\gamma_i m_{wx}(0, W_\ell, X_i) - \frac{1}{n_1} \sum_{A_\ell = 1} \sum_{J_i = \ell} m_{wx}(0, W_\ell, X_i).
\end{equation}
This bias is the post-weighting imbalance in a particular function of the cluster- and unit-level covariates: the conditional expectation function $m_{wx}(0, \cdot, \cdot)$. If we knew $m_{wx}$, then we could remove all bias by ensuring that this function is balanced. Unfortunately we do not know it, or else we would be able to perfectly impute the missing potential outcomes. Instead, following \citet{hirshberg2019minimax} and \citet{ben2021balancing}, we consider a class of potential models $\mathcal{M}_{wx}$ and consider the worst case bias over this class:
\[
  \imbal_{\mathcal{M}_{wx}}(\gamma) = \max_{m \in \mathcal{M}_{wx}} \left|\frac{1}{n_1}\sum_{A_\ell = 0} \sum_{J_i = \ell}\gamma_i m(0, W_\ell, X_i) - \frac{1}{n_1} \sum_{A_\ell = 1} \sum_{J_i = \ell} m(0, W_\ell, X_i) \right|.
\]
Here, since we are conditioning on both unit and cluster-level covariates, the model class $\mathcal{M}_{wx}$ can be relatively complex. For example, it might include interactions between $W_{J_i}$ and $X_i$.

Next, the design-conditional variance corresponds to the second term:
\begin{equation}
  \label{eq:var_unit}
    V^\text{unit} \equiv \Var\left(\hat{\mu}_0\left(\gamma\right) \mid \bm{W}, \bm{A}, \bm{J}, \bm{X} \right) = \frac{1}{n_1^2}\sum_{A_\ell = 0}\left[ \sum_{J_i = \ell} \gamma_i^2 \Var\left(\varepsilon_i\right) + \sum_{J_i = \ell}\sum_{J_k = \ell, k \neq i} \gamma_i\gamma_j\Cov(\varepsilon_i, \varepsilon_k)\right],
\end{equation}
where $\varepsilon_i = Y_i - m_{wx}(0, W_{J_i}, X_i)$ is the residual in the observed outcomes after conditioning on cluster- and unit-level covariates.
Note that due to potential correlation in the outcomes within clusters, the variance term includes a cross-term involving the product of pairs of weights on each unit, in addition to the typical squared weight term used in settings where the weights are independent. Taken together, we try to find weights that minimize an upper bound on the design-conditional mean square error:
\begin{equation}
  \label{eq:general_problem_unit}
  \min_{\gamma} \imbal_{\mathcal{M}_{wx}}(\gamma)^2 + \frac{1}{n_1^2}\sum_{A_\ell = 0}\left[ \sum_{J_i = \ell} \gamma_i^2 \Var\left(\varepsilon_i\right) + \sum_{J_i = \ell}\sum_{J_k = \ell, k \neq i} \gamma_i\gamma_j\Cov(\varepsilon_i, \varepsilon_k)\right].
\end{equation}

Implementing the optimization of Equations~\eqref{eq:general_problem_no_unit} requires making several modeling choices. First, we must choose the model class for the conditional expectation function $\mathcal{M}_{wx}$. Second, we must specify the variances and covariances of the residuals. In making these choices, we will attempt to strike a balance between flexibility and practicality.

We begin by selecting the model class for $\mathcal{M}_{wx}$. In the non-clustered treatment assignment setting, many model classes have been considered, from sparse models \citep{zubizarreta2015stable} to models with expanding basis functions \citep{wang2020minimal} to reproducing kernel Hilbert spaces \citep{hirshberg2019minimax}. See \citet{ben2021balancing} for a recent review and discussion on the difficulty of balancing these model classes. To target $\mu_{0wx}$ in the COS design, we consider a model class that is linear in a transformation of the covariates and incorporating both cluster- and unit-level covariates with $L^2$-bounded coefficients and an unbounded intercept, i.e.:
\[
    \mathcal{M}_{wx} = \Psi_{wx} \equiv \{\alpha + \beta \cdot \psi(w, x) \mid \|\beta\|_2 \leq C_{wx}, \alpha \in \R\}.\\
\]
These model classes naturally allow for a non-parametric extension to reproducing kernel hilbert spaces with infinite-dimensional transformations via the ``kernel'' trick, e.g., defining a kernel $k((w_1, x_1), (w_2, x_2)) = \psi(w_1, x_1) \cdot \psi(w_2, x_2)$.
Notably, these transformations, $\psi(w, x)$, can include interactions between the two levels of variables, allowing for the cluster context to affect the relationship between the outcome and the unit-level covariates.

To specify the variances and covariances of the residuals, we use a random effects model. Momentarily abusing notation, we write the residual between unit $i$'s outcome and its expected control outcome conditional on cluster- and unit-level covariates as $Y_i - m_{wx}(W_{J_i}, X_i) = \delta_{J_i} + \varepsilon_i$, with a cluster-level random effect $\delta_{J_i}$ and an independent unit-level residual $\varepsilon_i$. We parameterize the variance with two terms: $\Var(\delta_{J_i}) = \sigma^2 \rho$, and $\Var(\varepsilon_i) = \sigma^2 (1 - \rho)$. Under the random-effects model, $\sigma^2$ represents the total residual variance, and $\rho$ is the \emph{intra-class correlation} (ICC), which is a well-known measure of relatedness within clusters.
Under this random effects model, the variance is
\[
\frac{\sigma^2}{n_1^2}\sum_{A_\ell = 0} \left[(1-\rho)\sum_{J_i = \ell}\gamma_i^2 + \rho\left(\sum_{J_i = \ell}\gamma_i\right)^2\right].
\]
We then find the weights, $\gamma_i$ that solve the following optimization problem:
\begin{equation}
  \label{eq:balance_unit}
  \begin{aligned}
    \min_{\gamma} & \left\|\frac{1}{n_1}\sum_{A_\ell = 0} \sum_{J_i = \ell}\gamma_i \psi(W_\ell, X_i) - \frac{1}{n_1} \sum_{A_\ell = 1} \sum_{J_i = \ell} \psi(W_\ell, X_i)\right\|_2^2 + \frac{\sigma^2}{C_{wx}^2}\frac{1}{n_1^2}\sum_{A_\ell = 0} \left[(1-\rho)\sum_{J_i = \ell}\gamma_i^2 + \rho\left(\sum_{J_i = \ell}\gamma_i\right)^2\right]\\
    \text{subject to} & \sum_{A_\ell = 0}\sum_{J_i = \ell}\gamma_i = n_1 \;\; \text{ and } \;\; L \leq \bar{\gamma}_\ell \leq U,
  \end{aligned}
\end{equation}
where we have included some additional optional upper and lower bounds on the weights. This objective function implements the general balancing weights problem in Equation~\eqref{eq:general_problem_unit} with the model class $\Psi_{wx}$ and a random-effects model. The variance penalty includes two hyperparameters. First, there is the noise to signal ratio, $\sigma^2/C_w^2$, that measures the overall impact of the variance relative to the bias in the MSE. If this ratio is small, then better balance will be prioritized over lower variance; if the ratio is large then the opposite is true. The second hyperparameter is the ICC, which determines the level of penalization for each cluster. If the ICC is small and units' outcomes are nearly uncorrelated, then the weight on each unit will be penalized the same. Conversely, if the ICC is large then the outcomes are very correlated within clusters, the penalization focuses on the total weight assigned to each cluster instead.  We discuss setting these hyperparameters below. Note that this specification of the variance penalty is equivalent to \citet{Rubinstein2022_region}, who consider region-level policy analysis via weighting.

The weights are also constrained to sum to the total number of treated units, and to be bounded between a lower bound $L$ and an upper bound $U$. The former constraint comes from the possibility of an unbounded intercept in the model class $\Phi_{wx}$; by ensuring the sum constraint we can be sure that the estimator is invariant to constant shifts in the outcome. The latter constraint acts as a form of regularization, and allows us to perform typically post-hoc adjustments such as weight trimming directly when finding our weights. If $L = 0$ and $U = \infty$, the estimator will be restricted from extrapolating away from the support of the control data \citep[see][for further discussion on the role of extrapolation]{benmichael2021_augsynth}. If we additionally set $U$ to be non-infinite, we can prevent any weights from becoming extremely large, at the cost of balance.

\subsection{Conditioning on cluster-level covariates only}
\label{sec:bal_no_unit}

Next, we consider estimating $\mu_{0w}$, which does not condition on unit-level covariates. This weighting estimator ignores the unit-level information, with weights that are constant within clusters, i.e. $\gamma^\text{clus}_i = \bar{\gamma}_{J_i}$ for all units $i$. Here, the results are a special case of the CUD, removing the unit-level covariates from the optimization procedure. The design-conditional bias and variance in this special case are:
\begin{align}
  \label{eq:bias_no_unit}
  \E\left[ \hat{\mu}_0\left(\gamma^\text{clus}\right) - \tilde{\mu}_{0w} \mid \bm{W}, \bm{A}, \bm{J} \right] & = \frac{1}{n_1}\sum_{A_\ell = 0}n_\ell \bar{\gamma}_\ell m_w(0, W_\ell) - \frac{1}{n_1} \sum_{A_\ell = 1} n_\ell m_w(0, W_\ell),\\
  \label{eq:var_no_unit}
  V^\text{clus} \equiv \Var\left(\hat{\mu}_0\left(\gamma^\text{clus}\right) \mid \bm{W}, \bm{A}, \bm{J}\right) & = \frac{1}{n_1^2} \sum_{A_\ell = 0} \bar{\gamma}_\ell^2 \Var\left(\sum_{J_i = \ell} e_i \mid \bm{J}\right).
\end{align}
Here, the bias only depends on imbalance in a function of the cluster-level covariates, $m_w$. As above, we consider a model class $\mathcal{M}_w$ and upper bound the bias by worst-case imbalance:
\[
  \imbal_{\mathcal{M}_w}(\bar{\gamma}) \equiv \max_{m \in \mathcal{M}_w} \left| \frac{1}{n_1}\sum_{A_\ell = 0}n_\ell \bar{\gamma}_\ell m(0, W_\ell) - \frac{1}{n_1} \sum_{A_\ell = 1} n_\ell m(0, W_\ell) \right|.
\]
Comparing the variances $V^\text{clus}$ and $V^\text{unit}$, we see that restricting to cluster-level covariates only removes the cross term between weights in the same cluster. So $V^\text{clus}$ depends principally on the sum of the squared weights, weighted by the total variance of the outcomes in each cluster. Tailoring Equation~\eqref{eq:general_problem_unit} yields the following optimization problem to control the MSE:
\begin{equation}
  \label{eq:general_problem_no_unit}
  \min_{\bar{\gamma}} \imbal_{\mathcal{M}_w}(\bar{\gamma})^2 + \frac{1}{n_1^2} \sum_{A_\ell = 0} \bar{\gamma}_\ell^2 \Var\left(\sum_{J_i = \ell} e_i \mid \bm{J}\right).
\end{equation}

To implement this optimization procedure, we will use the same model class and variance model as above, removing the unit-level covariates.
For the model class, we use the set of models that are linear in transformations of the cluster-level covariates:
\[
    \mathcal{M}_w = \Phi_w  \equiv \{\alpha + \beta \cdot \phi(w) \mid \|\beta\|_2 \leq C_w, \alpha \in \R\}.
\]
For the variance model, we again
assume that the residual decomposes into a cluster-level residual and an independent unit-level residual: $Y_i - m_w(0, W_{J_i}) = e_i + d_{J_i}$, with $\Var(e_i) = s^2 (1 - r)$ and $\Var(d_\ell) = s^2 r$. Here, $s^2$ represents the total variance in the residual, and $r$ is the ICC. Now, the variance can be written as
\[
  \frac{s^2}{n_1^2} \sum_{A_\ell = 0} \bar{\gamma}_\ell^2 ((1 - r)n_\ell + r n_\ell^2).
\]
Putting together the pieces, we find weights $\bar{\gamma}$ that solve the following optimization problem
\begin{equation}
  \label{eq:balance_nounit}
  \begin{aligned}
    \min_{\bar{\gamma}} \; & \left\|\frac{1}{n_1}\sum_{A_\ell = 0}n_\ell \bar{\gamma}_\ell \phi(W_\ell) - \frac{1}{n_1} \sum_{A_\ell = 1} n_\ell \phi(W_\ell)\right\|_2^2 + \frac{s^2}{C_w^2}\frac{1}{n_1^2} \sum_{A_\ell = 0} \bar{\gamma}_\ell^2 ((1 - r)n_\ell + r n_\ell^2)\\
    \text{subject to } & \sum_\ell\sum_{A_\ell = 0}n_\ell \bar{\gamma}_\ell = n_1 \;\;\text{and } \;\; L \leq \bar{\gamma}_\ell \leq U.
  \end{aligned}
\end{equation}

The cluster-level balancing problem in Equation~\eqref{eq:balance_nounit} and the unit-level problem in Equation~\eqref{eq:balance_unit} share many similarities, including the constraints on the weights and the noise-to-signal ratio.\footnote{The noise-to-signal ratio will in general be different when unit-level covariates are included.} However, the key difference is whether the weights can differ within clusters in order to balance the additional unit-level covariates. This is most apparent in how the ICC affects the variance penalty. When only cluster-level covariates are included, the ICC only appears in the role that the cluster size $n_\ell$ has in the variance penalty. When weights are allowed to differ by unit, however, this creates more options in how the weights are penalized.

\subsection{Hyperparameter selection}
\label{sec:hyper}

Both of the objective functions specified above include two terms that we consider hyperparameters: the overall noise to signal ratio, ($s^2 / C^2_{w}$ or $\sigma^2/C_{wx}^2$) and the ICC. The noise to signal ratio governs the bias-variance tradeoff. As it approaches zero, more and more emphasis is placed on achieving covariate balance; in the (unreasonable) noiseless limit, we would not need or want to penalize the variance. On the other hand, as the noise to signal ratio increases and the covariates are less predictive of the outcome, the optimization will prioritize variance more. In contrast, the ICC term determines the form of the variance penalty. In both cases, when the ICC is 0, the weights on each unit are penalized separately since units' outcomes are independent. When the ICC is 1, units' outcomes completely move together within a group, in which case the total weight placed on the group is penalized. Intermediate values of the ICC correspond to a mixture of these penalization schemes.

Note that the particular values of the noise to signal ratio and the ICC only enter the estimator through the variance penalties on the weights in Equations  \eqref{eq:balance_unit} and \eqref{eq:balance_nounit}. This is why we view them as hyperparameters in the optimization problem rather than parameters to be estimated. However, we can use the data at hand to guide our choice of these hyperparameters. We do so by regressing the outcome on the covariates with a random intercept model. The random intercept model estimates cluster- and unit-level variance terms that can be used as an estimate for the ICC. To estimate the signal-to-noise ratio, we take the ratio of the estimated residual variance and the squared sum of the estimated regression coefficients from the fitted model. Below, we use this heuristic to choose the hyperparameters in our simulation studies and empirical analyses. Note that this procedure can induce a dependence between the outcomes and the weights through the hyper-parameters, which can be avoided via sample splitting. Other forms of hyperparameter selection are possible, including cross-validation style approaches that evaluate balance on a held-out sample \citep{wang2020minimal}.

\subsection{When should we include unit-level covariates in the Cluster-Only Design?}
\label{sec:should_we}

Given that under the COD it is sufficient to only include cluster-level covariates, one natural question for this design is whether it is useful to also include unit-level covariates. Including unit-level covariates may decrease the variance: balancing these covariates is akin to adjusting for baseline covariates that predict the outcome in randomized experiments. On the other hand, balancing unit-level covariates may lead to more extreme weights, and, depending on the ICC, this could lead to higher variance.

To characterize this trade off, we consider the ratio of the design conditional variance with and without including unit-level covariates, $V^\text{unit}$ and $V^\text{clus}$, respectively. In cases with excellent balance in both cluster- and unit-level covariates, this will give a measure of the expected difference in the overall mean squared error. To simplify the problem, we restrict our attention to the cluster-level random effects model, and assume that the ICC is unchanged when conditioning on unit-level covariates (i.e., $\rho = r$). Under these assumptions, the variance ratio is
\begin{equation}
  \label{eq:var_ratio}
  \frac{V^\text{unit}}{V^\text{clus}} = \frac{\sigma^2}{s^2} \ d_\text{eff}(\rho),
\end{equation}
where
\[
  d_\text{eff}(\rho) = \frac{(1-\rho) \sum_{A_\ell = 0} \sum_{J_i = \ell}\gamma_i^2 + \rho\sum_{A_\ell = 0}\left(\sum_{J_i = \ell}\gamma_i\right)^2}{(1 - \rho) \sum_{A_\ell = 0} n_\ell  \bar{\gamma}_\ell^2 + \rho \sum_{A_\ell = 0} \left(n_\ell \bar{\gamma}_\ell\right)^2}
\]
is the design effect of including unit-level covariates.

The first term, $\sigma^2/s^2 < 1$ represents the reduction in total variance in the residuals by conditioning on unit-level covariates, and serves to reduce the variance ratio by incorporating additional information. The more predictive the unit-level variables are, the lower we expect this ratio to be. Counteracting this, we have the design effect $d_\text{eff}(\rho)$.
Rearranging Equation~\eqref{eq:var_ratio}, we see that in order for it to be beneficial to include unit-level covariates, the variance after conditioning on unit-level covariates must decrease by at least the design effect, $\sigma^2 < s^2 d_\text{eff}(\rho)$.

For any given value of $\rho$, we expect the cluster-level weights to be less extreme, because they do not have to additionally balance the unit-level covariates, so typically $d_\text{eff}(\rho) \geq 1$. Whether the sum of the squared weights on the units or on the clusters matters more to the design effect depends on the ICC. When $\rho$ is small, there is little effect of clustering on the variance, and so we have a standard tradeoff: including unit-level covariates will increase efficiency as long as the effective sample size does not decrease by more than $1 - \sigma^2/s^2$. Conversely, when $\rho$ is large, the sum of the squared weights on the clusters will dominate, in which case the unit-level weights can differ substantially from their average cluster-level weight without incurring much extra variance. In this case including unit-level covariates can improve efficiency even if they are only somewhat predictive. \citet{hedges2007intraclass} reported that from 41 clustered randomized experiments in education, the ICCs range from 0.07 to 0.31, with an average value of 0.17. \citet{Small:2008b} report that ICCs in the range of .002 to 0.03 in public health interventions that target clusters such as hospitals or clinics. Thus, in education settings --- where ICCs are larger, and we typically have strongly predictive covariates such as prior test scores --- it will likely be beneficial to include unit-level covariates. In contrast, for public health settings it may depend much more on the particular context and the set of unit-level covariates available.

\section{Extensions and Inference}
\label{sec:ext}
\subsection{Dealing with poor balance}

When there are a large number of cluster-level covariates relative to the total number of clusters --- or conversely, when there is a small number of clusters overall --- it may be difficult to find weights that achieve good balance. Often this occurs when the overlap between the treated and control cluster covariate distributions is limited \citep{keele2022overlap}. This is true for the application in Section~\ref{sec:catholic} below. We consider two different approaches to account for this. First, we outline using an outcome model to correct for the bias due to remaining imbalance. Second, we consider changing the estimand by also weighting the treated units to find a maximally sized overlapping set between the treated and control units.

\subsubsection{Bias correction with an outcome estimator}

Bias correction, sometimes also called augmentation, is a popular approach to reducing bias due to imbalance.
We describe the bias-correction procedure when only including both cluster- and unit-level covariates; restricting to cluster-level covariates alone will be analogous.
First we estimate the conditional expectation function to get estimates $\hat{m}_{wx}(0, W_{J_i}, X_i)$.
There are many possible estimation strategies, one choice being regularized regression.
Then we perform bias-correction by estimating $\mu_{0}$ as
\begin{equation}
  \label{eq:bias_correct}
  \hat{\mu}^\text{bc}_0(\gamma) \equiv \hat{\mu}_0(\gamma) + \frac{1}{n_1} \sum_{A_\ell = 1} \sum_{J_i = \ell} \hat{m}_{wx}(0, W_\ell, X_i) - \frac{1}{n_1}\sum_{A_\ell = 0} \sum_{J_i = \ell}\gamma_i \hat{m}_{wx}(0, W_\ell, X_i).
\end{equation}
If we compare this to the bias in Equation~\eqref{eq:bias_var_unit}, we see that $\hat{\mu}^\text{bc}_0(\hat{\gamma})$ uses the estimated model to estimate the bias due to imbalance after weighting. Then it attempts to remove the bias by subtracting the estimated bias off of the estimate.
In the bias-variance decomposition above, bias correction changes the imbalance measure to be over the worst-case \emph{model error} $\hat{m} - m$, which will generally be smaller than the imbalance in the worst-case model \citep{hirshberg2021augmented}.
This bias correction is akin to what has been proposed for matching estimators \citep{Abadie2011_bias_match} and augmented IPW \citep{Robins1994_aipw}, and has been used with a variety of balancing weights estimators \citep[e.g.][]{athey2018approximate, benmichael2021_multical}.

\subsubsection{Subset Weights: Finding a maximally overlapping set}

One issue with bias correction based on an outcome model is that the additional use of modeling can lead to extrapolation away from the control units' data \citep{benmichael2021_augsynth}. One alternative to outcome modeling that can be attractive, especially if overlap is limited, is to focus on a more limited estimand than the ATT. Specifically, one such estimand is the average treatment effect for the overlap population (ATO). The ATO corresponds to the treatment effect for the marginal population that might or might not receive the treatment of interest rather than a known, a priori well-defined population such as the treated group. \citet{li2018balancing} developed model-based overlap weights for the ATO that continuously down-weight the units in the tails of the propensity score distribution. We develop an analog with balancing weights for the COS setting.

To do so, we will weight the treated units as well as the control units, leading to an estimator
\begin{equation}
  \label{eq:w_effect}
  \hat{\tau}(\gamma) = \frac{1}{n_1} \sum_{A_\ell = 1} \sum_{J_i = \ell} \gamma_i Y_i - \frac{1}{n_0} \sum_{A_\ell = 0} \sum_{J_i = \ell} \gamma_i Y_i.
\end{equation}
By weighting the treated units as well as the control units, we are no longer estimating the expected effect among the treated units. We are instead trimming the estimand.
To understand what the new estimand is, consider the design conditional expectation:
\[
  \E\left[\hat{\tau}(\gamma) \mid \bm{W}, \bm{J}, \bm{X}\right] = \frac{1}{n_1} \sum_{A_\ell = 1} \sum_{J_i = \ell} \gamma_i m_{wx}(0, W_\ell, X_i)- \frac{1}{n_0} \sum_{A_\ell = 0} \sum_{J_i = \ell} \gamma_i m_{wx}(0, W_\ell, X_i) + \frac{1}{n_1}\sum_{A_\ell = 1}\sum_{J_i = \ell}\gamma_i \tau(W_\ell, X_i),
\]
where $\tau(w, x) = \E[Y_i(1, \bm{J}) - Y_i(0, \bm{J}) \mid W_{J_i} = w, X_i = x]$ is the conditional average treatment effect (CATE). As before, we want to find weights that balance the conditional expected control outcome. Then $\hat{\tau}(\gamma)$ will be an unbiased estimator of a weighted average of conditional treatment effects for the treated group. We will once again try to find weights that minimize the worst case balance across a model class $\mathcal{M}_{wx}$, now also weighting the treated units:
\[
    \imbal_{\mathcal{M}_{wx}}^o(\gamma) = \max_{m \in \mathcal{M}_{wx}} \left|\frac{1}{n_1}\sum_{A_\ell = 0} \sum_{J_i = \ell}\gamma_i m(0, W_\ell, X_i) - \frac{1}{n_0} \sum_{A_\ell = 1} \sum_{J_i = \ell} \gamma_i m(0, W_\ell, X_i) \right|.
\]
The design-conditional variance of $\hat{\tau}(\gamma)$ is
\[
  \Var\left(\hat{\tau}(\gamma) \mid \bm{W}, \bm{A}, \bm{J}, \bm{X} \right) = \sum_{\ell=1}^m\left(\frac{A_\ell}{n_1} + \frac{1 - A_\ell}{n_0}\right)^2\left[ \sum_{J_i = \ell} \gamma_i^2 \Var\left(\varepsilon_i\right) + \sum_{J_i = \ell}\sum_{J_k = \ell, k \neq i} \gamma_i\gamma_j\Cov(\varepsilon_i, \varepsilon_k)  \right],
\]
where we generalize the definition of the residual to be $\varepsilon_i = Y_i - m_{wx}(A_{J_i}, W_{J_i}, X_i)$.
With this, we once again try to find weights that minimize the imbalance and the variance. Focusing on the specialization to the constrained linear transformation model class $\Psi_{wx}$ and the random effects variance model, we find weights $\gamma_i$ that solve:
\begin{equation}
  \label{eq:balance_overlap_unit}
  \begin{aligned}
    \min_{\gamma} & \left\|\frac{1}{n_0}\sum_{A_\ell = 0} \sum_{J_i = \ell}\gamma_i \psi(W_\ell, X_i) - \frac{1}{n_1} \sum_{A_\ell = 1} \sum_{J_i = \ell} \gamma_i\psi(W_\ell, X_i)\right\|_2^2\\
    &  + \frac{\sigma^2}{C_{wx}^2}\sum_{\ell=1}^m \left(\frac{A_\ell}{n_1} + \frac{1 - A_\ell}{n_0}\right)^2\left[(1-\rho)\sum_{J_i = \ell}\gamma_i^2 + \rho\left(\sum_{J_i = \ell}\gamma_i\right)^2\right],\\
    \text{subject to} & \sum_{A_\ell = 0}\sum_{J_i = \ell}\gamma_i = n_0 \;\; \text{ and } \sum_{A_\ell = 1}\sum_{J_i = \ell}\gamma_i = n_1 \;\; \text{ and } \;\; L \leq \bar{\gamma}_\ell \leq U.
  \end{aligned}
\end{equation}
This optimization problem tries to find weights on both treated and control units such that the weighted average of their transformed covariates are similar. Hereafter, we refer to these weights as subset weights, since they weight the subset of the data for which the covariate distributions overlap.
The variance penalty in Equation~\eqref{eq:balance_overlap_unit} serves to ensure that this weighted subset is pushed towards having a larger effective number of units. Finally, Equation~\eqref{eq:balance_overlap_unit} is also related to the Lagrangian dual of the SVM problem. \citet{Tarr2021_svm} show that the SVM dual minimizes the same measure of imbalance, but includes a penalty to ensure that the sum of the weights is not small, leading to a different measure of the ``size'' of the overlapping set.

\subsection{Variance Estimation and Uncertainty Quantification}
\label{sec:var_estimate}

We now turn to constructing asymptotically valid confidence intervals for $\mu_0$ by estimating the variance of $\hat{\mu}_0(\hat{\gamma}) -\mu_0$ and relying on asymptotic normality. When estimating the variance, it is important to account for dependence within clusters. One method for variance estimation in this context is the cluster-robust sandwich estimator \citep{huber1967under,white1980}, which we adapt here. Focusing on the CUD first, we estimate the conditional expectation function $\hat{m}_{wx}(0, w, x)$, then compute unit-level residuals $\hat{\varepsilon}_i \equiv Y_i - \hat{m}_{wx}(0, W_{J_i}, X_i)$.
This leads to a plug-in estimator for the variance:
\[
  \hat{V}^\text{unit} \equiv \frac{1}{n_1^2}\sum_{A_\ell = 0}\left[ \sum_{J_i = \ell} \gamma_i^2 \hat{\varepsilon}_i^2 + \sum_{J_i = \ell}\sum_{J_k = \ell, k \neq i} \gamma_i\gamma_j\hat{\varepsilon}_i \hat{\varepsilon}_k\right].
\]
\noindent Assumption~\ref{a:technical} in the Supplementary Materials lists regularity conditions for asymptotic inference, based on conditions from \citet{Hansen2019_asymptotic}. Importantly, these regularity conditions allow the cluster sizes to grow, but limit the number of clusters. 

Here we highlight the dependence of the variances on the sample size by indexing them by $n$.
To state the asymptotic normality result below, define $\mu_l \equiv \sum_{J_i = \ell}\sum_{J_k = \ell}\hat{\gamma}_i \hat{\gamma}_k \E\left[\varepsilon_i\hat{m}(0, W_\ell, X_i) \mid  \bm{W}, \bm{A}, \bm{J}, \bm{X} \right]$ as the covariance between the true residuals and the estimated model predictions in cluster $\ell$.
\begin{theorem}
  \label{thm:asymp_normal}
  Under the regularity conditions in Assumption~\ref{a:technical}, with constant lower and upper bounds $L$ and $U$ in Equation~\eqref{eq:balance_unit}, if $\imbal_{\mathcal{M}_{wx}}(\hat{\gamma}) = o_p\left(1/\sqrt{V^\text{unit}_n}\right)$, then
  \[
    \frac{1}{\sqrt{V^\text{unit}_n}} \left(\hat{\mu}(\hat{\gamma}) - \tilde{\mu}_{0wx}\right) \Rightarrow N(0,1).
  \]
  Furthermore, if $\sup_{w,x} |\hat{m}(0, w, x) - m(0, w, x)| = o_p(1)$ and $\frac{1}{n_1^2}\sum_{A_\ell = 0} \mu_l \to 0$
  then $\frac{\hat{V}^\text{unit}} {V^\text{unit}_n} \to 1$ in probability, and consequently,
  \[
    \frac{1}{\sqrt{\hat{V}^\text{unit}_n}} \left(\hat{\mu}(\hat{\gamma}) - \tilde{\mu}_{0wx}\right) \Rightarrow N(0,1).
  \]
\end{theorem}
Theorem~\ref{thm:asymp_normal} implies that we can construct approximate $1-\alpha$ confidence intervals as $\hat{\mu}_0(\hat{\gamma}) \pm z_{1-\alpha / 2} \sqrt{\hat{V}^\text{unit}}$, where $z_{1-\alpha/2}$ is the $1-\alpha$ quantile of the standard normal distribution.
The key assumption is that the weights $\hat{\gamma}$ can achieve good balance in the sense that the worst-case imbalance $\imbal_{\mathcal{M}_{wx}}(\hat{\gamma})$ converges to zero faster than $1/\sqrt{V^\text{unit}}$. This rate depends on the correlation within clusters. If units are uncorrelated, it will scale with the total number of units; if units are perfectly correlated then it will scale with the number of clusters; see \citet{Hansen2019_asymptotic} for further discussion on rates of convergence with clustered data. Whether the weights can actually achieve this level of balance depends on the model class. \citet{hirshberg2019minimax} show that for Reproducing Kernel Hilbert Spaces with i.i.d. data, the imbalance will be small enough, while \citet{hirshberg2021augmented} show that the bias-corrected estimator will have small enough bias in more general settings. See \citet{ben2021balancing} for further discussion. Finally, Theorem~\ref{thm:asymp_normal} assumes that the estimated model $\hat{m}$ is consistent and that the covariance between the true residuals and the estimated model predictions converges to zero, ensuring that the variance estimator is consistent. The latter can be guaranteed by using sample-splitting or cross fitting, or by using models with restricted complexity.

Below, we estimate $\hat{m}$ via regularized weighted least squares on the control units, using the weights from Equation \eqref{eq:balance_nounit} above. If we only include an intercept and exclude all of the covariates in the regression, then these variance estimates are equivalent to the cluster-robust sandwich estimates of a weighted mean. However, this approach may be conservative as it disregards the variance reduction due to balancing the covariates.

Finally, by including a separate regression for the treated units, we can also compute a plug in estimate of the variance for the maximal overlap estimator $\hat{\tau}(\gamma)$ above. In addition, to construct a variance estimate for the COD we can follow the same procedure: (i) estimate the conditional expectation function $\hat{m}_{w}(0, w)$; (ii) estimate the residuals $\hat{\varepsilon}_i \equiv Y_i - \hat{m}_{w}(0, W_{J_i})$; and (iii) create a plug-in estimate for the variance:
\[
  \hat{V}^\text{clus} \equiv \frac{1}{n_1^2} \sum_{A_\ell = 0} \bar{\gamma}_\ell^2 \sum_{J_i = \ell} \sum_{J_k = \ell}\hat{e}_i\hat{e}_k.
\]

\section{Simulation Study}
\label{sec:sim}

Next, we conduct a simulation study to understand the performance of balancing weights using a simulation design from \citet{keele2021matching} developed to evaluate multilevel matching. First, we describe the data generation process (DGP), which we alter slightly to manipulate the level of overlap between the treated and control distributions. The DGP partially depends on empirical data from a summer school reading intervention that follows the COS template. In the data, there are 18 treated schools with 1,367 students, and 26 control schools with 2,060 students. There are 5 student-level variables and 9 school-level variables.  The student-level variables are reading and math test scores and indicator variables for race, ethnicity, and sex. The school-level covariates include the percentage of students who receive free/reduced price lunch, who are English language learners, and who are proficient in math and reading based on state standardized tests. The school-level covariates also include the share of teachers who are novice (e.g., in their first year), the rate of year-to-year staff turnover, and student average daily attendance. 

For this DGP, we first fit a school-level propensity score model where we regressed the observed treatment indicator on the following set of school level variables: the percentage of students receiving free/reduced price lunch, the percentage of students who are English language learners, the percentage of teachers who are novice, and student average daily attendance. We denote this estimated propensity score as $\hat{e}(w)$.  We define the latent probability of treatment as:
\[
Z^* = (\hat{e}(w)/c) + \text{Unif}(-.5,.5)
\]
where $c$ controls the level of overlap between treated and control clusters. Observed treatment status is generated via the following model: $Z_j = 1(Z^* > 0.25)$.

Next, we fit an outcome model for the observed data. Here, we regressed reading scores on the student-level covariates with the basis expanded to include interactions between race and test scores. After model fitting, we save $\hat{\beta_0}$, the intercept from this regression. We use $\tau$ to denote the true treatment effect estimate in the simulation, and set it to be 0.3 of a standard deviation of the raw outcome measure. We denote student-level reading scores with $R_{ij}$, student-level math scores as $M_{ij}$, and the percentage of students proficient in math and reading in each school with $P_{j}$. Next, we generate potential outcomes under control as:
\[
y_0 = \hat{\beta_0} + 2.5 R_{ij} + 2.5 M_{ij} + 1.9 P_{j} + v_1,
\]
\noindent where $v_1$ is a draw from a normal distribution that is mean zero with a standard deviation of 12. We selected these parameter values so that misspecification would produce a bias of approximately 0.3 standard deviations on a standardized scale. A bias of this magnitude is large enough to completely obscure the true treatment effect. Next, we generated potential outcomes under treatment as $y_1 = y_0 + \tau$, and we generated simulated outcomes as $\tilde{Y}_{ij} = Z_j y_1 + (1 - Z_j) y_0$. Note that in this DGP, selection into treatment at the school level is only a function of school-level covariates, but the potential outcomes under control are a function of the student-level test scores and, to a lesser extent, the overall quality of the school as measured by the percentage of proficient students. This DGP also induces a correlation within clusters. The average ICC across simulations was 0.29 with a standard deviation of 0.04. We judge this to be a common data structure in educational COS designs. Using this DGP, we conduct two different simulation studies.

\subsection{Simulation study 1}

In the first study, we conduct a comparative analysis of adjustment methods for COS designs. First, we produce a naive estimate of the treatment effect as the difference in means without any covariate adjustment, and an estimate using multilevel matching as implemented in \citet{Keele:2016b}. Next, we implement our balancing weights in two ways: first, as balancing weights by solving Equation~\eqref{eq:balance_unit} and second, as subset balancing weights by solving Equation~\eqref{eq:balance_overlap_unit}. We set the hyperparameters via the heuristic in Section \ref{sec:hyper}, using the estimated coefficients and variance components from a multilevel regression model. In this simulation, we focus on how the performance changes as we vary the level of overlap by setting $c$ to values of 1, 2.5, 7.5, and 10. When $c=1$ this will induce poor overlap, and when $c=10$ there is excellent overlap, with intermediate values increasing the level of overlap. For each scenario, we repeated the simulation 1,000 times, and we report the bias and root mean-squared error (RMSE). In our results, we standardize the bias, dividing it by the standard deviation of the control group's outcomes from the original data.

In the first panel of Figure~\ref{fig:sim}, we plot the bias as a function of overlap for all four estimation methods. We observe that matching and both types of weights reduce bias compared to an unadjusted estimate. However, when overlap is poor, both sets of weights remove substantially more bias than multilevel matching. The performance of matching could be improved by trimming treated observations. However, with the balancing weights, we can reduce the bias compared to matching without having to change the estimand. The subset weights also alter the estimand, but allow for an unbiased treatment effect estimate.

\begin{figure}
  \centering
    \includegraphics[scale=0.6]{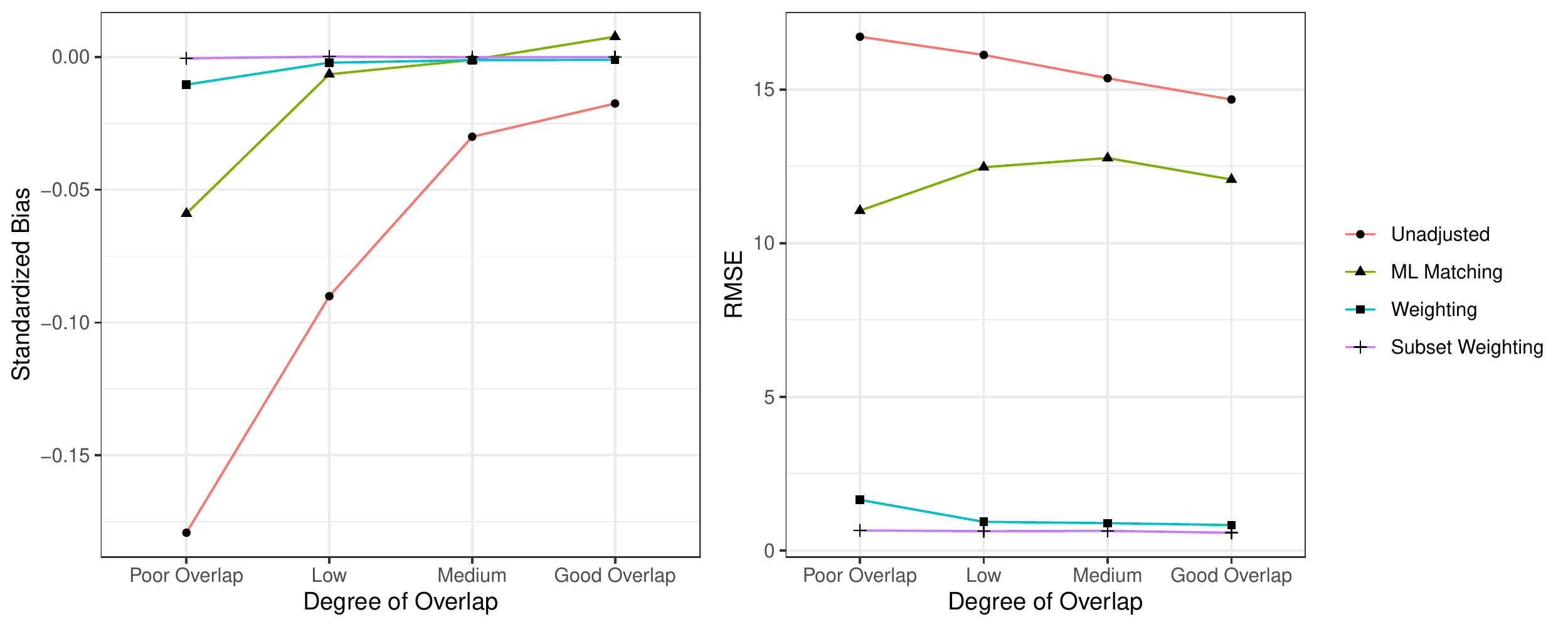}
    \caption{Bias and RMSE for matching and two different weighting estimators by overlap condition.}
  \label{fig:sim}
\end{figure}

We next compare the methods in terms of RMSE. Both matching and weighting discard data. Matching discards some of the control schools and units, while weighting gives some of the control schools and students zero weight. One open question is whether one method or the other does so in a more efficient fashion. In the second panel of Figure~\ref{fig:sim}, we plot the RMSE as a function of overlap for all the estimation methods. In this scenario, the difference in performance between matching and balancing weights is clear. While the subset weights have the lowest RMSE, the difference between the two weighting methods is minor. However, both forms of balancing weights outperform matching across all overlap scenarios by nearly a factor of 10. Balancing weights thus manage to reduce bias while also retaining a much larger effective sample size, which improves efficiency. Overall, we find that balancing weights are a clear improvement over matching. Balancing weights have lower bias than matching when overlap is poor and also are considerably more efficient.

\subsection{Simulation study 2}

In the second simulation, we focus on the performance of the proposed plug-in variance estimator, relative to using the standard weighted cluster-robust sandwich estimator, for the balancing weights solving Equation~\eqref{eq:balance_unit}. In this simulation, we fixed the overlap parameter at 10 and vary the number of clusters. We control the number of clusters by resampling clusters with replacement from the original data, and then generate outcomes and treatments following the DGP above. We used cluster sample sizes of 50, 100, 150, 200 and 250. For each scenario, we repeated the simulation 1,000 times, and we report the average standard error estimate and the length of the 95\% confidence interval.

Figure~\ref{fig:sim3} shows the results for the good overlap scenario. In the first panel, we compare the magnitude of the two variance estimates. As we expect, our proposed plug-in method produces smaller standard errors by removing variance due to the covariates. The difference between the two methods is largest when the number of clusters is small. These smaller standard errors also produce narrower confidence intervals. We also measured nominal coverage of the confidence intervals for both methods and found that both methods led to over-coverage of the true effect. This result is consistent with the general finding that the sandwich variance estimator for weighting methods is conservative. The pattern from the good overlap scenario is very similar for the other two overlap settings, so we report those results in the supplementary materials.

\begin{figure}
  \centering
    \includegraphics[scale=0.6]{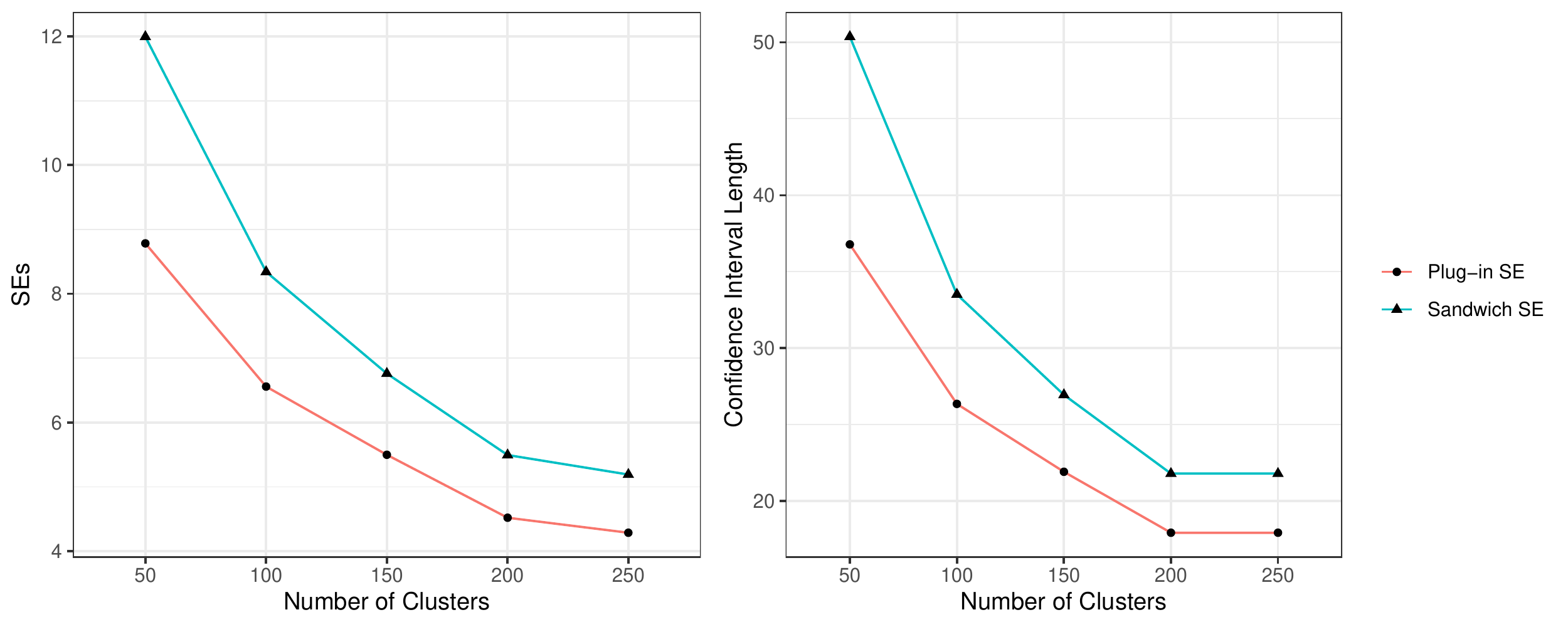}
    \caption{Comparative performance for two different variance estimators.}
  \label{fig:sim3}
\end{figure}

\section{Applications}
\label{sec:apps}

Next, we present results from two different empirical applications. The first, from education, compares the performance of Catholic and public schools. The overlap between Catholic and public schools is known to be poor \citep{keele2022overlap}, which allows us to study the performance of the subset weights in a context for which they were designed. The second application, from health services research, compares different residency programs for surgeons. Here, overlap is much better, but the sample sizes may prove computationally challenging for multilevel matching.

\subsection{Catholic Schools}
\label{sec:catholic}

Our analysis is a replication \citet{keele2022overlap}, which used multilevel matching and highlighted the limited amount of overlap between Catholic and public schools. Here, we compare results based on our proposed weighting methods to those based on multilevel matching. The data are a public release of the 1982 High School and Beyond survey and includes records for 7,185 high school students from 160 schools.  Of these schools, 70 are Catholic schools and are thus considered treated in this application, while the remainder are public high schools and thus serve as a reservoir of controls. The average number of students sampled per Catholic school is approximately 50 and ranges from 20 -- 67, while the average number of students sampled per public school is 41 and ranges from 14 -- 61. The data contain covariates on both students and schools. Student-level covariates are: an indicator for whether or not the student is female; an indicator for whether a student belongs to a particular racial/ethnic group; and a scale for socioeconomic status (SES). Three of the school-level measures are school-level averages of these student-level measures. Three additional school-level covariates are: total enrollment; the percentage of students on an academic track; and a measure of disciplinary climate. The disciplinary climate variable is a composite measure created from a factor score on measures of the number of attacks on teachers, fights, and other disciplinary incidents. This variable ranges from -1.7 to 2.7.

In our analysis, we computed COS balancing weights for the ATT and subset weights for the overlapping set. Using a random effects model, we estimated the two components of the hyperparameter: the estimated ICC is 0.036, and the estimated signal-to-noise ratio is 1.2. We also include two of the multilevel matches implemented in \citet{keele2022overlap}. The first of these matches retains all the Catholic schools, and the second trimmed 10 Catholic schools to improve the balance and increase overlap. Table~\ref{tab:bal} contains a comparison of how well each method balanced the baseline covariates as measured by the standardized difference: the difference in weighted/matched Catholic and public schools means divided by the pooled standard deviation before adjustment. The lack of overlap is apparent in the size of the standardized differences in the school-level covariates; several of the standardized differences are larger than 0.50, and three exceed 1. We observe that balancing weights outperform matching in terms of balance. While the standardized differences for the COS balancing weights are still fairly large for two covariates, they are less than half of those obtained via multilevel matching. Moreover, if we are willing to alter the estimand, the subset weights are able to nearly exactly balance the Catholic and public school distributions.

\begin{table}[htbp]
\centering
\begin{threeparttable}
    \caption{Balance Table Comparing Catholic and Public Schools on Baseline Covariate Distributions.}
    \label{tab:bal}
\begin{tabular}{lccccc}
  \toprule
 & Unweighted & Balancing & Subset & Matching & Matching -- \\
 &  & Weights & Weights &  & Trimmed  \\
  \midrule
  Student SES & 0.48 & 0.00 & 0.00 & 0.33 & 0.13 \\
  \% Students Minority & -0.14 & 0.26 & 0.00 & 0.24 & 0.22 \\
  \% Students Female & -0.00 & -0.10 & -0.00 & 0.04 & -0.09 \\
  Enrollment & -0.80 & -0.17 & -0.01 & -0.58 & -0.66 \\
  \% Students on Academic Track & 1.52 & 0.32 & 0.00 & 1.27 & 0.90 \\
  Disciplinary Climate Scale & -1.64 & -0.47 & -0.01 & -0.92 & -0.92 \\
  School SES Average & 1.17 & 0.15 & 0.00 & 0.79 & 0.31 \\
   \bottomrule
\end{tabular}
\begin{tablenotes}[para]
Note:  Cell entries are standardized differences, the difference in means divided by the pool standard
deviation.
\end{tablenotes}
\end{threeparttable}
\end{table}

To better understand the differences between the COS weights and the COS subset weights, we provide descriptive statistics on the largest weights. For the balancing weights, the vast majority of the weights are between 0 and 1, but there are 63 weights that are larger than 10, and the largest weight has a value of 92. For the COS subset weights the largest weight is 14. This should reduce the likelihood of unstable behavior in the treatment effect estimates due to large weights. Next, we measure the effective sample sizes to understand the loss of information due to weighting. In our data, before adjustment there are 5,273 students. The effective sample size for the balancing weights is 1,710, and 1,036 for the subset weights. For comparison, the effective sample size for the match with all treated schools is 570, and 392 for the match that trimmed treated schools.

Next, we review the point estimates for the Catholic school effect. Note that the outcome is a standardized test score, so estimates are measured in standard deviations.  First, the unadjusted effect is 0.43 with a 95\% CI of (0.31, 0.55). Next, the Catholic school effect estimated by balancing weights is -0.08. The 95\% confidence interval using the sandwich variance estimator is (-0.76, 0.59), and the confidence interval based on our proposed plug-in estimator is (-0.47, 0.31). The Catholic school estimate based on subset weights is 0.07. The 95\% confidence interval using the sandwich variance estimator is (-0.29, 0.44), and the confidence interval based on our proposed plug-in estimator is -0.08--0.24. Note that these estimates are not directly comparable, as they target different estimands, but we observe in both cases small point estimates with confidence intervals that include zero. Notably, the plug-in variance estimator produces shorter confidence intervals consistent with the simulation results. For comparison, we next report the estimates based on multilevel matching. The estimate based on multilevel matching is 0.26 (95\% CI: 0.12, 0.40), and the estimate based on matching with trimming is 0.13 (95\% CI: -0.06, 0.31). The estimates based on matching are closer to those based on weighting if we include additional bias reduction via outcome modeling. This suggests that weighting was much more successful than matching at removing bias without reference to outcomes.

Finally, while the subset balancing weights were able to balance the data better even with poor overlap, we had to change the estimand. To aid in the interpretation of this estimand, we plot the covariate means for Catholic and public schools before and after subset weighting in Figure~\ref{fig:estimand}. That is, we compare the raw Catholic school means to the weighted Catholic school means, and the same for public schools. We observe that Catholic and public school do not differ much in terms of gender and racial mix. While Catholic and public schools differ some in terms of SES, the key difference is in terms of disciplinary climate.  For this covariate, we observe the largest difference between the weighted and unweighted estimates for both Catholic and public schools. In the overlap population of schools, Catholic schools have a much stricter disciplinary climate and public schools have a more permissive disciplinary climate.

\begin{figure}
  \centering
    \includegraphics[scale=0.5]{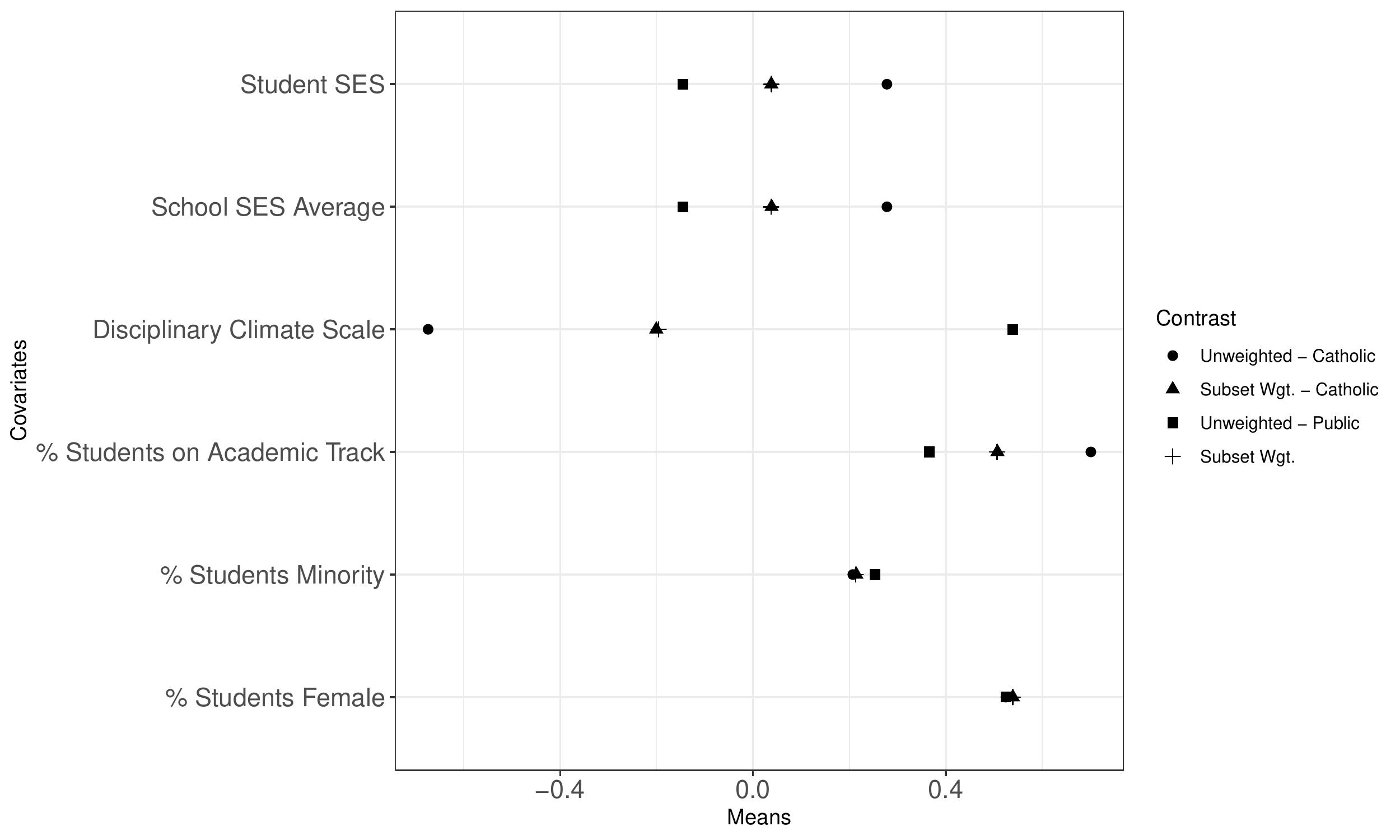}
    \caption{Covariate means for Catholic and Public schools before and after subset weighting.}
  \label{fig:estimand}
\end{figure}

\subsection{Surgical Training}

One important question in health services research is whether certain aspects of surgical training have an effect on patient outcomes \citep{asch2009evaluating,bansal2016using,zaheer2017comparing,sullivan2012effect}. \citet{sellers2018association} studied whether surgeons from university-based residency programs produce superior patient outcomes than surgeons trained in community-based residency programs.  The original study used data based on all-payer hospital discharge claims from New York, Florida and Pennsylvania from 2012--2013. In the data, surgeons were classified as having attended a university-based (UBR) or non-university based (NUBR) residency based on the residency program listed in the American Medical Association Masterfile.
The data contain covariates for surgeons, including age, sex, and years of training completion, and covariates for patients, such as sociodemographic and clinical characteristics including 31 comorbidities based on Elixhauser indices \citep{elixhauser1998comorbidity}. The primary outcome was postoperative complications. Complications were identified using ICD-9 diagnosis codes and collapsed into a binary variable indicating the development of 1 or more complications. They compared surgeon performance between UBR and NUBR surgeons for patients that underwent one of 44 common operations performed by general surgeons in an inpatient setting \citep{sellers2018association}.

In this application, there are 498 treated surgeons and 1,201 control surgeons. Overall, there are 86,305 patients operated on by UBR surgeons, and 193,307 patients operated on by NUBR surgeons. The number of patients treated by each surgeon varied from five to 1,074 over the two-year period. In the UBR application, standardized differences before weighting are relatively small--one indication that overlap is good for this data set. As such, we only target the ATT estimand and do not use the subset weights. We again set hyperparameter values based on estimates from a random effects model. In this data, the estimated ICC is 0.016, and the noise to signal ratio is 0.221. Figure~\ref{fig:ubr} contains a balance plot for the subset of covariates with the largest imbalances.
For each of these covariates, weighting improves any imbalance relative to the unadjusted data, giving close to exact balance.

In the full data, there are 279,611 patients. After weighting, the effective sample size is 262,672. Here, the loss of sample size is relatively small. We also attempted to implement a multilevel match  on a desktop computer with 32 GB of RAM. However, we received an error that \texttt{R} was unable to allocate enough memory to complete the match. For the balancing weights, we are able to estimate weights is less than a minute. As such, balancing weights have clear computational advantages for COS applications with larger sample sizes.
		
\begin{figure}
  \centering
    \includegraphics[scale=0.7]{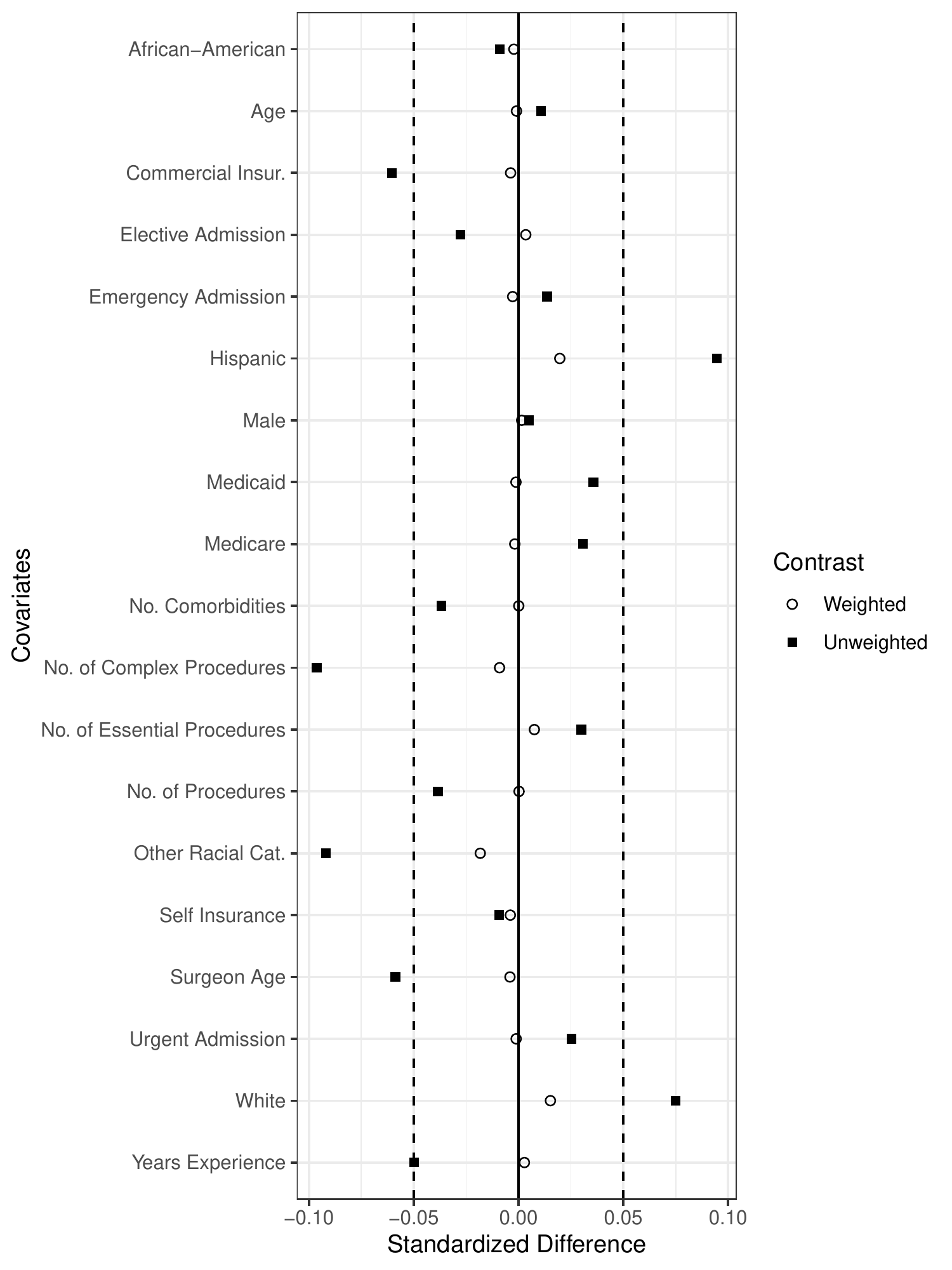}
    \caption{Balance Plot: UBR vs NUBR surgeons for the set of covariates with the largest baseline imbalances. }
  \label{fig:ubr}
\end{figure}

Next, we estimate treatment effects. First, we estimate the unadjusted treatment effect via a regression model with clustering at the surgeon level. We find that for UBR surgeons, the estimated percentage of cases with a complication is 1.35 percentage points lower, and the 95\% confidence interval does not contain zero (95\% CI: -2.14, -0.55). Once we account for confounding via weighting, the estimate treatment effect is -0.85 percentage points. We estimated 95\% confidence intervals using both the cluster-robust sandwich variance estimator (95\% CI: -1.62, -0.09) and our proposed plug-in variance estimator (95\% CI: -1.53, -0.18). Because the balance is excellent, further bias reduction via an outcome model does not meaningfully change the estimate. In sum, we find that UBR surgeons do indeed appear to cause fewer complications, even when accounting for observed confounding. However, the magnitude of the treatment effect is smaller once we control for observed covariates.

\section{Conclusion}
\label{sec:dis}

We introduced an approximate balancing weight estimator for designs where treatments are administered to entire clusters such as school or hospitals. To do so, we considered two potential estimands --- one adjusting for both unit-level and cluster-level covariates and another adjusting for cluster-level covariates alone --- that have different identification assumptions associated with them.
For both of these estimands, we find weights to minimize an upper bound on the mean square error of the resulting weighting estimator. When adjusting for cluster-level covariates alone, the weights are constant within clusters, while the weights vary across units within clusters when including unit-level covariates. This affects the overall variance, as units' outcomes can be correlated within clusters. We showed that when it is sufficient to adjust only for cluster-level covariates in order to estimate the ATT, there can be efficiency gains to including unit-level covariates as well, depending on the predictive strength of those covariates and the level of correlation between units' outcomes within clusters. We also considered two adaptations to deal with cases where it is impossible to find weights that achieve good covariate balance: (i) bias-correction via an outcome model and (ii) changing the estimand and finding an overlapping weighted subset of the data. We then showed how to construct confidence intervals that are asymptotically valid under certain conditions. In a simulation study, we demonstrated that our proposed weighting estimator outperformed multilevel matching. This was especially true in terms of using more of the sample size which resulted in higher efficiency and lower RMSE. In two empirical applications, we found balancing weights also had several practical advantages over multilevel matching.

There are several avenues for future work. First, while we propose a heuristic for choosing the hyperparameters from the observed data, an important question is how to choose these hyperparameters in a rigorous, data-driven way that keeps weights independent of outcomes. Second, we can consider resampling approaches to uncertainty quantification such as the block weighted bootstrap proposed by \citet{Cui2022_match}. Finally, many COS designs have a longitudinal structure, where data is available at multiple granularities over time. For example, many policy changes happen at the state-level, but county and municipality-level data series exist for the outcome of interest. Exploring these cases and extending our analysis to such settings will be important areas for future research.

\clearpage
\renewcommand{\refname}{Bibliography \& References Cited}
\bibliography{cluster}

\begin{thebibliography}{}

\bibitem[\protect\citeauthoryear{Abadie and Imbens}{Abadie and
  Imbens}{2011}]{Abadie2011_bias_match}
Abadie, A. and G.~W. Imbens (2011).
\newblock {Bias-corrected matching estimators for average treatment effects}.
\newblock {\em Journal of Business and Economic Statistics\/}~{\em 29\/}(1),
  1--11.

\bibitem[\protect\citeauthoryear{Asch, Nicholson, Srinivas, Herrin, and
  Epstein}{Asch et~al.}{2009}]{asch2009evaluating}
Asch, D.~A., S.~Nicholson, S.~Srinivas, J.~Herrin, and A.~J. Epstein (2009).
\newblock Evaluating obstetrical residency programs using patient outcomes.
\newblock {\em Jama\/}~{\em 302\/}(12), 1277--1283.

\bibitem[\protect\citeauthoryear{Athey, Imbens, and Wager}{Athey
  et~al.}{2018}]{athey2018approximate}
Athey, S., G.~W. Imbens, and S.~Wager (2018).
\newblock Approximate residual balancing: debiased inference of average
  treatment effects in high dimensions.
\newblock {\em Journal of the Royal Statistical Society: Series B (Statistical
  Methodology)\/}~{\em 80\/}(4), 597--623.

\bibitem[\protect\citeauthoryear{Bansal, Simmons, Epstein, Morris, and
  Kelz}{Bansal et~al.}{2016}]{bansal2016using}
Bansal, N., K.~D. Simmons, A.~J. Epstein, J.~B. Morris, and R.~R. Kelz (2016).
\newblock Using patient outcomes to evaluate general surgery residency program
  performance.
\newblock {\em JAMA surgery\/}~{\em 151\/}(2), 111--119.

\bibitem[\protect\citeauthoryear{Ben-Michael, Feller, and Hartman}{Ben-Michael
  et~al.}{2021}]{benmichael2021_multical}
Ben-Michael, E., A.~Feller, and E.~Hartman (2021).
\newblock {Multilevel calibration weighting for survey data}.

\bibitem[\protect\citeauthoryear{Ben-Michael, Feller, Hirshberg, and
  Zubizarreta}{Ben-Michael et~al.}{2021}]{ben2021balancing}
Ben-Michael, E., A.~Feller, D.~A. Hirshberg, and J.~R. Zubizarreta (2021).
\newblock The balancing act in causal inference.
\newblock {\em arXiv preprint arXiv:2110.14831\/}.

\bibitem[\protect\citeauthoryear{Ben-Michael, Feller, and
  Rothstein}{Ben-Michael et~al.}{2021}]{benmichael2021_augsynth}
Ben-Michael, E., A.~Feller, and J.~Rothstein (2021).
\newblock {The Augmented Synthetic Control Method}.
\newblock {\em Journal of the American Statistical Association\/}~{\em
  116\/}(536), 1789--1803.

\bibitem[\protect\citeauthoryear{Chattopadhyay, {Christopher H. Hase}, and
  Zubizarreta}{Chattopadhyay et~al.}{2020}]{Chattopadhyay2020}
Chattopadhyay, A., {Christopher H. Hase}, and J.~R. Zubizarreta (2020).
\newblock {Balancing Versus Modeling Approaches to Weighting in Practice}.
\newblock {\em Statistics in Medicine\/}~{\em in press}.

\bibitem[\protect\citeauthoryear{Cochran}{Cochran}{1965}]{Cochran:1965}
Cochran, W.~G. (1965).
\newblock The planning of observational studies of human populations.
\newblock {\em Journal of Royal Statistical Society, Series A\/}~{\em
  128\/}(2), 234--265.

\bibitem[\protect\citeauthoryear{Cui, Yang, Reich, and Gill}{Cui
  et~al.}{2022}]{Cui2022_match}
Cui, C., S.~Yang, B.~J. Reich, and D.~A. Gill (2022).
\newblock {Matching Estimators of Causal Effects in Clustered Observational
  Studies with Application to Quantifying the Impact of Marine Protected Areas
  on Biodiversity}.

\bibitem[\protect\citeauthoryear{Elixhauser, Steiner, Harris, and
  Coffey}{Elixhauser et~al.}{1998}]{elixhauser1998comorbidity}
Elixhauser, A., C.~Steiner, D.~R. Harris, and R.~M. Coffey (1998).
\newblock Comorbidity measures for use with administrative data.
\newblock {\em Medical care\/}~{\em 36\/}(1), 8--27.

\bibitem[\protect\citeauthoryear{Hainmueller}{Hainmueller}{2011}]{Hainmueller2011}
Hainmueller, J. (2011).
\newblock {Entropy Balancing for Causal Effects: A Multivariate Reweighting
  Method to Produce Balanced Samples in Observational Studies}.
\newblock {\em Political Analysis\/}~{\em 20}, 25--46.

\bibitem[\protect\citeauthoryear{Hansen and Lee}{Hansen and
  Lee}{2019}]{Hansen2019_asymptotic}
Hansen, B.~E. and S.~Lee (2019).
\newblock {Asymptotic theory for clustered samples}.
\newblock {\em Journal of Econometrics\/}~{\em 210\/}(2), 268--290.

\bibitem[\protect\citeauthoryear{Hazlett}{Hazlett}{2019}]{Hazlett2019}
Hazlett, C. (2019).
\newblock {Kernel Balancing : A flexible non-parametric weighting procedure for
  estimating causal effects}.
\newblock {\em Statistica Sincia\/}.

\bibitem[\protect\citeauthoryear{Hedges and Hedberg}{Hedges and
  Hedberg}{2007}]{hedges2007intraclass}
Hedges, L.~V. and E.~C. Hedberg (2007).
\newblock Intraclass correlation values for planning group-randomized trials in
  education.
\newblock {\em Educational Evaluation and Policy Analysis\/}~{\em 29\/}(1),
  60--87.

\bibitem[\protect\citeauthoryear{Hirshberg, Maleki, and Zubizarreta}{Hirshberg
  et~al.}{2019}]{hirshberg2019minimax}
Hirshberg, D.~A., A.~Maleki, and J.~Zubizarreta (2019).
\newblock Minimax linear estimation of the retargeted mean.
\newblock {\em arXiv preprint arXiv:1901.10296\/}.

\bibitem[\protect\citeauthoryear{Hirshberg and Wager}{Hirshberg and
  Wager}{2021}]{hirshberg2021augmented}
Hirshberg, D.~A. and S.~Wager (2021).
\newblock Augmented minimax linear estimation.
\newblock {\em The Annals of Statistics\/}~{\em 49\/}(6), 3206--3227.

\bibitem[\protect\citeauthoryear{Huber}{Huber}{1967}]{huber1967under}
Huber, P.~J. (1967).
\newblock Under nonstandard conditions.
\newblock In {\em Proceedings of the Fifth Berkeley Symposium on Mathematical
  Statistics and Probability: Weather Modification; University of California
  Press: Berkeley, CA, USA}, pp.\  221.

\bibitem[\protect\citeauthoryear{Keele, Lenard, and Page}{Keele
  et~al.}{2021}]{keele2021matching}
Keele, L., M.~Lenard, and L.~Page (2021).
\newblock Matching methods for clustered observational studies in education.
\newblock {\em Journal of Research on Educational Effectiveness\/}~{\em
  14\/}(3), 696--725.

\bibitem[\protect\citeauthoryear{Keele, Lenard, and Page}{Keele
  et~al.}{2022}]{keele2022overlap}
Keele, L., M.~Lenard, and L.~Page (2022).
\newblock Overlap violations in clustered observational studies of educational
  interventions.
\newblock {\em Journal of Research on Educational Effectiveness\/}, 1--18.

\bibitem[\protect\citeauthoryear{Keele and Zubizarreta}{Keele and
  Zubizarreta}{2017}]{Keele:2015}
Keele, L.~J. and J.~Zubizarreta (2017).
\newblock Optimal multilevel matching in clustered observational studies: A
  case study of the effectiveness of private schools under a large-scale
  voucher system.
\newblock {\em Journal of the American Statistical Association\/}~{\em
  112\/}(518), 547--560.

\bibitem[\protect\citeauthoryear{Li, Morgan, and Zaslavsky}{Li
  et~al.}{2018}]{li2018balancing}
Li, F., K.~L. Morgan, and A.~M. Zaslavsky (2018).
\newblock Balancing covariates via propensity score weighting.
\newblock {\em Journal of the American Statistical Association\/}~{\em
  113\/}(521), 390--400.

\bibitem[\protect\citeauthoryear{Page, Lenard, and Keele}{Page
  et~al.}{2020}]{pagedesign2019}
Page, L.~C., M.~Lenard, and L.~Keele (2020, July-Sept).
\newblock The design of clustered observational studies in education.
\newblock {\em AERA Open\/}~{\em 6\/}(3), 1--14.

\bibitem[\protect\citeauthoryear{Pimentel, Page, Lenard, and Keele}{Pimentel
  et~al.}{2018}]{Keele:2016b}
Pimentel, S.~D., L.~C. Page, M.~Lenard, and L.~J. Keele (2018).
\newblock Optimal multilevel matching using network flows: An application to a
  summer reading intervention.
\newblock {\em Annals of Applied Statistics\/}~{\em 12\/}(3), 1479--1505.

\bibitem[\protect\citeauthoryear{Raudenbush}{Raudenbush}{1997}]{raudenbush1997statistical}
Raudenbush, S.~W. (1997).
\newblock Statistical analysis and optimal design for cluster randomized
  trials.
\newblock {\em Psychological Methods\/}~{\em 2\/}(2), 173.

\bibitem[\protect\citeauthoryear{Robins, Rotnitzky, and {Ping Zhao}}{Robins
  et~al.}{1994}]{Robins1994_aipw}
Robins, J.~M., A.~Rotnitzky, and L.~{Ping Zhao} (1994).
\newblock {Estimation of Regression Coefficients When Some Regressors are not
  Always Observed}.
\newblock {\em Journal of the American Statistical Association\/}~{\em 89427},
  846--866.

\bibitem[\protect\citeauthoryear{Rubin}{Rubin}{1974}]{Rubin:1974}
Rubin, D.~B. (1974).
\newblock Estimating causal effects of treatments in randomized and
  nonrandomized studies.
\newblock {\em Journal of Educational Psychology\/}~{\em 6\/}(5), 688--701.

\bibitem[\protect\citeauthoryear{Rubin}{Rubin}{2007}]{Rubin:2007}
Rubin, D.~B. (2007).
\newblock The design versus the analysis of observational studies for causal
  effects: parallels with the design of randomized trials.
\newblock {\em Statistics in medicine\/}~{\em 26\/}(1), 20--36.

\bibitem[\protect\citeauthoryear{Rubin}{Rubin}{2008}]{Rubin:2008}
Rubin, D.~B. (2008, September).
\newblock For objective causal inference, design trumps analysis.
\newblock {\em The Annals of Applied Statistics\/}~{\em 2\/}(3), 808--840.

\bibitem[\protect\citeauthoryear{Rubinstein, Haviland, and Choi}{Rubinstein
  et~al.}{2022}]{Rubinstein2022_region}
Rubinstein, M., A.~Haviland, and D.~Choi (2022).
\newblock {Balancing weights for estimated region-level data: the effect of
  Medicaid Expansion on the uninsurance rate among states that did not expand
  Medicaid}.

\bibitem[\protect\citeauthoryear{Sellers, Keele, Sharoky, Wirtalla, Bailey, and
  Kelz}{Sellers et~al.}{2018}]{sellers2018association}
Sellers, M.~M., L.~J. Keele, C.~E. Sharoky, C.~Wirtalla, E.~A. Bailey, and
  R.~R. Kelz (2018).
\newblock Association of surgical practice patterns and clinical outcomes with
  surgeon training in university-or nonuniversity-based residency program.
\newblock {\em JAMA surgery\/}~{\em 153\/}(5), 418--425.

\bibitem[\protect\citeauthoryear{Small, Have, and Rosenbaum}{Small
  et~al.}{2008}]{Small:2008b}
Small, D.~S., T.~R.~T. Have, and P.~R. Rosenbaum (2008, March).
\newblock Randomization inference in a group--randomized trial of treatments
  for depression: Covariate adjustment, noncompliance, and quantile effects.
\newblock {\em Journal of the American Statistical Association\/}~{\em
  103\/}(481), 271--279.

\bibitem[\protect\citeauthoryear{Sullivan, Sue, Bucholz, Yeo, Bell~Jr, Roman,
  and Sosa}{Sullivan et~al.}{2012}]{sullivan2012effect}
Sullivan, M.~C., G.~Sue, E.~Bucholz, H.~Yeo, R.~H. Bell~Jr, S.~A. Roman, and
  J.~A. Sosa (2012).
\newblock Effect of program type on the training experiences of 248 university,
  community, and us military-based general surgery residencies.
\newblock {\em Journal of the American College of Surgeons\/}~{\em 214\/}(1),
  53--60.

\bibitem[\protect\citeauthoryear{Tarr and Imai}{Tarr and
  Imai}{2021}]{Tarr2021_svm}
Tarr, A. and K.~Imai (2021).
\newblock {Estimating Average Treatment Effects with Support Vector Machines}.

\bibitem[\protect\citeauthoryear{Wang and Zubizarreta}{Wang and
  Zubizarreta}{2019}]{Wang2018}
Wang, Y. and J.~R. Zubizarreta (2019).
\newblock {Minimal dispersion approximately balancing weights: asymptotic
  properties and practical considerations}.
\newblock {\em Biometrika\/}.

\bibitem[\protect\citeauthoryear{Wang and Zubizarreta}{Wang and
  Zubizarreta}{2020}]{wang2020minimal}
Wang, Y. and J.~R. Zubizarreta (2020).
\newblock Minimal dispersion approximately balancing weights: asymptotic
  properties and practical considerations.
\newblock {\em Biometrika\/}~{\em 107\/}(1), 93--105.

\bibitem[\protect\citeauthoryear{White}{White}{1980}]{white1980}
White, H. (1980).
\newblock A heteroskedasticity-consistent covariance matrix estimator and a
  direct test for heteroskedasticity.
\newblock {\em Econometrica: Journal of the Econometric Society\/}, 817--838.

\bibitem[\protect\citeauthoryear{Ye, Westling, Page, and Keele}{Ye
  et~al.}{2022}]{ye2022clustered}
Ye, T., T.~Westling, L.~Page, and L.~Keele (2022).
\newblock Nonparametric identification of causal effects in clustered
  observational studies with differential selection.
\newblock Unpublished Manuscript.

\bibitem[\protect\citeauthoryear{Zaheer, Pimentel, Simmons, Kuo, Datta,
  Williams, Fraker, and Kelz}{Zaheer et~al.}{2017}]{zaheer2017comparing}
Zaheer, S., S.~D. Pimentel, K.~D. Simmons, L.~E. Kuo, J.~Datta, N.~Williams,
  D.~L. Fraker, and R.~R. Kelz (2017).
\newblock Comparing international and united states undergraduate medical
  education and surgical outcomes using a refined balance matching methodology.
\newblock {\em Annals of surgery\/}~{\em 265\/}(5), 916--922.

\bibitem[\protect\citeauthoryear{Zhao}{Zhao}{2019}]{Zhao2019}
Zhao, Q. (2019).
\newblock {Covariate balancing propensity score by tailored loss functions}.
\newblock {\em Annals of Statistics\/}~{\em 47\/}(2), 965--993.

\bibitem[\protect\citeauthoryear{Zhao and Percival}{Zhao and
  Percival}{2016}]{Zhao2016a}
Zhao, Q. and D.~Percival (2016).
\newblock {Entropy Balancing is Doubly Robust}.
\newblock {\em Journal of Causal Inference\/}.

\bibitem[\protect\citeauthoryear{Zubizarreta}{Zubizarreta}{2015}]{zubizarreta2015stable}
Zubizarreta, J.~R. (2015).
\newblock Stable weights that balance covariates for estimation with incomplete
  outcome data.
\newblock {\em Journal of the American Statistical Association\/}~{\em
  110\/}(511), 910--922.

\end{thebibliography}

\clearpage

\noindent	

{\bfseries\large
Supplement to ``Approximate Balancing Weights for Clustered Observational Study Designs'}

\section{Proofs and derivations}
\label{sec:proofs}

\begin{assumption}
  \label{a:technical}
  For some $2 \leq r < \infty$,
  \begin{enumerate}[label ={(\roman*)}]
    \item \label{a:cluster_size_1} $\frac{1}{n_0}\left(\sum_{A_\ell = 0}n_\ell^r\right)^{\frac{2}{r}} \leq C < \infty$
    \item \label{a:cluster_size_2}$\max_{\ell} \frac{n_\ell^2}{n_0} \to 0$
    \item \label{a:regularity} $\lim_{M \to \infty} \sup_i \E\left[|\varepsilon_1|^r \bbone\{|\varepsilon_i| > M\}\right] = 0$
    \item \label{a:positive_var} There exists some $\lambda > 0$ such that $V^\text{unit}_n \geq \lambda$
    \item As $n \to \infty$, $\frac{n_1}{n} \to p$ for some constant $p$
  \end{enumerate}
\end{assumption}

\begin{proof}[Proof of Theorem \ref{thm:asymp_normal}]
  First, we decompose the estimation error scaled by the inverse standard deviation as
  \begin{align*}
    \frac{1}{\sqrt{V^\text{unit}_n}} \left(\hat{\mu}(\hat{\gamma}) - \tilde{\mu}_{0wx}\right) & = \underbrace{\frac{1}{\sqrt{V^\text{unit}_n}} \left(\frac{1}{n_1}\sum_{A_\ell = 0} \sum_{J_i = \ell}\gamma_i m_{wx}(0, W_\ell, X_i) - \frac{1}{n_1} \sum_{A_\ell = 1} \sum_{J_i = \ell} m_{wx}(0, W_\ell, X_i)\right)}_{o_p(1)}\\
    & + \frac{1}{\sqrt{V^\text{unit}_n}} \frac{1}{n_1}\sum_{A_\ell = 0}\sum_{J_i = \ell} \hat{\gamma}_i \varepsilon_i.
  \end{align*}
  Because $\imbalance_{\mathcal{M}_{wx}}(\hat{\gamma}) = o_p\left(1/\sqrt{V^\text{unit}_n}\right)$, the first term is $o_p(1)$, so we focus on the second term.

  Note that because each weight is bounded, $L \leq \hat\gamma_i \leq U$, the product of the weights and the residuals also satisfy Assumption~\ref{a:technical}\ref{a:regularity}, i.e. 
  \[
    \lim_{M \to \infty} \sup_i \E\left[|\hat{\gamma}_i|^r|\varepsilon_1|^r \bbone\{|\varepsilon_i| > M\}\right] \leq \lim_{M \to \infty} \sup_i \E\left[\max\{|L|, U\}|\varepsilon_1|^r \bbone\left\{|\varepsilon_i| > \frac{M}{\max{|L|, U}}\right\}\right] = 0.
  \]
  Now define $\Sigma_n = \frac{n_1^2}{n_0} V_n^\text{unit}$. By, \citet{Hansen2019_asymptotic} Theorem 2,
  \[
    \frac{1}{\sqrt{V^\text{unit}_n}} \frac{1}{n_1}\sum_{A_\ell = 0}\sum_{J_i = \ell} \hat{\gamma}_i \varepsilon_i = \frac{\sqrt{n_0}}{\sqrt{\Sigma_n}} \frac{1}{n_0}\sum_{A_\ell = 0}\sum_{J_i = \ell} \hat{\gamma}_i \varepsilon_i \Rightarrow N(0,1).
  \]
  Putting these two terms together and applying Slutsky's theorem, we have that
  \[
    \frac{1}{\sqrt{V^\text{unit}_n}} \left(\hat{\mu}(\hat{\gamma}) - \tilde{\mu}_{0wx}\right) \Rightarrow N(0,1).
  \]

  We now turn to the second claim that $\frac{\hat{V}^\text{unit}} {V^\text{unit}_n} \to 1$ in probability, or equivalently, that 
  $\frac{\hat{V}_n^\text{unit} - V^\text{unit}_n} {V^\text{unit}_n} = o_p(1)$. Define the variance estimate with the true residuals as
  \[
    \tilde{V}^\text{unit}_n \equiv  \frac{1}{n_1^2}\sum_{A_\ell = 0}\left[ \sum_{J_i = \ell} \gamma_i^2 \varepsilon_i^2 + \sum_{J_i = \ell}\sum_{J_k = \ell, k \neq i} \gamma_i\gamma_j\varepsilon_i\varepsilon_k\right].
  \]
  Adding and subtracting $\tilde{V}^\text{unit}_n$, we get that 
  \[
    \frac{\hat{V}_n^\text{unit} - V^\text{unit}_n} {V^\text{unit}_n}  = \frac{\hat{V}^\text{unit} - \tilde{V}^\text{unit}_n} {V^\text{unit}_n} + \frac{\tilde{V}^\text{unit}_n - V^\text{unit}_n} {V^\text{unit}_n}
  \]
  Theorem 3 in \citet{Hansen2019_asymptotic} implies that  $\frac{\tilde{V}^\text{unit}_n - V^\text{unit}_n} {V^\text{unit}_n} - o_p(1)$. Since $V_n^\text{unit} \geq \lambda > 0$ by Assumption~\ref{a:technical}\ref{a:positive_var}, showing that $\hat{V}^\text{unit}_n - V^\text{unit}_n = o_p(1)$ will complete the claim. We can write this difference as
  \begin{align}
    &\hat{V}^\text{unit}_n - V^\text{unit}_n \nonumber\\
    = &\frac{1}{n_1^2}\sum_{A_\ell = 0} \sum_{i \mid J_i = \ell}\sum_{k \mid J_k = \ell} \hat{\gamma}_i\hat{\gamma}_k (m(0, W_\ell, X_i) - \hat{m}(0, W_\ell, X_i)) (m(0, W_\ell, X_k) - \hat{m}(0, W_\ell, X_k))
    \label{eq:part1}\\
    + & \frac{1}{n_1^2}\sum_{A_\ell = 0} \sum_{i \mid J_i = \ell}\sum_{k \mid J_k = \ell} \hat{\gamma}_i\hat{\gamma}_k \varepsilon_i (m(0, W_\ell, X_k) - \hat{m}(0, W_\ell, X_k))
    \label{eq:part2}.
  \end{align}
  By H\"{o}lder's inequality, \eqref{eq:part1} is bounded by
  \[
    \eqref{eq:part1} \leq  \left(\max_{\ell, i} |m(0, W_\ell, X_i) - \hat{m}(0, W_\ell, X_i)| \right)^2\frac{1}{n_1^2}\sum_{A_\ell = 0} \sum_{i \mid J_i = \ell}\sum_{k \mid J_k = \ell} |\hat{\gamma}_i||\hat{\gamma}_k|.
  \]
  Since $\sup_{w,x} |\hat{m}(0, w, x) - m(0, w, x)| = o_p(1)$, this bound implies that \eqref{eq:part1} is also $o_p(1)$.

  To control \eqref{eq:part2}, define
  \[
    \mu_\ell \equiv \sum_{J_i = \ell}\sum_{J_k = \ell}\hat{\gamma}_i \hat{\gamma}_k \E\left[\varepsilon_i\hat{m}(0, W_\ell, X_i) \mid  \bm{W}, \bm{A}, \bm{J}, \bm{X} \right],
  \]
  so that the overall expectation of \eqref{eq:part2} is
  \[
    \E[\eqref{eq:part2} \mid \bm{W}, \bm{A}, \bm{J}, \bm{X} ] = \frac{1}{n_1^2}\sum_{A_\ell = 0} \mu_\ell,
  \]
  since $\E\left[\varepsilon_i m(0, W_\ell, X_i) \mid  \bm{W}, \bm{A}, \bm{J}, \bm{X} \right] = 0$.
  Now we bound the variance as
  \begin{align*}
    \Var\left(\eqref{eq:part2} \mid \bm{W}, \bm{A}, \bm{J}, \bm{X} \right) & = \frac{1}{n_1^4}\sum_{A_\ell = 0} \Var\left( \sum_{J_i = \ell}\sum_{J_k = \ell} \hat{\gamma}_i \hat{\gamma}_k \varepsilon_i (m(0, W_\ell, X_k) - \hat{m}(0, W_\ell, X_k)) \mid \bm{W}, \bm{A}, \bm{J}, \bm{X} \right)\\
    & \leq \left(\sup_{w,x} |\hat{m}(0, w, x) - m(0, w, x)|\right)^2\frac{1}{n_1^4}\sum_{A_\ell = 0} \Var\left( \sum_{J_i = \ell}\sum_{J_k = \ell} \hat{\gamma}_i \hat{\gamma}_k \varepsilon_i \mid \bm{W}, \bm{A}, \bm{J}, \bm{X} \right).
  \end{align*}

  So by Chebyshev's inequality, $\eqref{eq:part2} - \frac{1}{n_1^2}\sum_{A_\ell = 0} \mu_\ell = o_p(1)$. Finally, by assumption, $\lim_{n\to\infty} \frac{1}{n_1^2}\sum_{A_\ell = 0} \mu_\ell  = 0$, so \eqref{eq:part2} is $o_p(1)$.
  Putting together the pieces, we have that $\frac{\hat{V}_n^\text{unit} - V^\text{unit}_n} {V^\text{unit}_n} = o_p(1)$.

  Finally, again by Slutsky's theorem, $\frac{1}{\sqrt{\hat{V}^\text{unit}_n}} \left(\hat{\mu}(\hat{\gamma}) - \tilde{\mu}_{0wx}\right) \Rightarrow N(0,1)$.

\end{proof}

\section{Simulation Study: Additional Results}

Next, we report the full set of results from the second simulation study. Specifically, we report the results from the medium and poor overlap scenarios for the simulation that compared variance estimators. Figure~\ref{fig:se-med-overlap} contains the simulation results in the good overlap scenario. Figure~\ref{fig:se-poor-overlap} contains the simulation results in the poor overlap scenario.

\begin{figure}[htbp]
  \centering
   \includegraphics[scale=0.7]{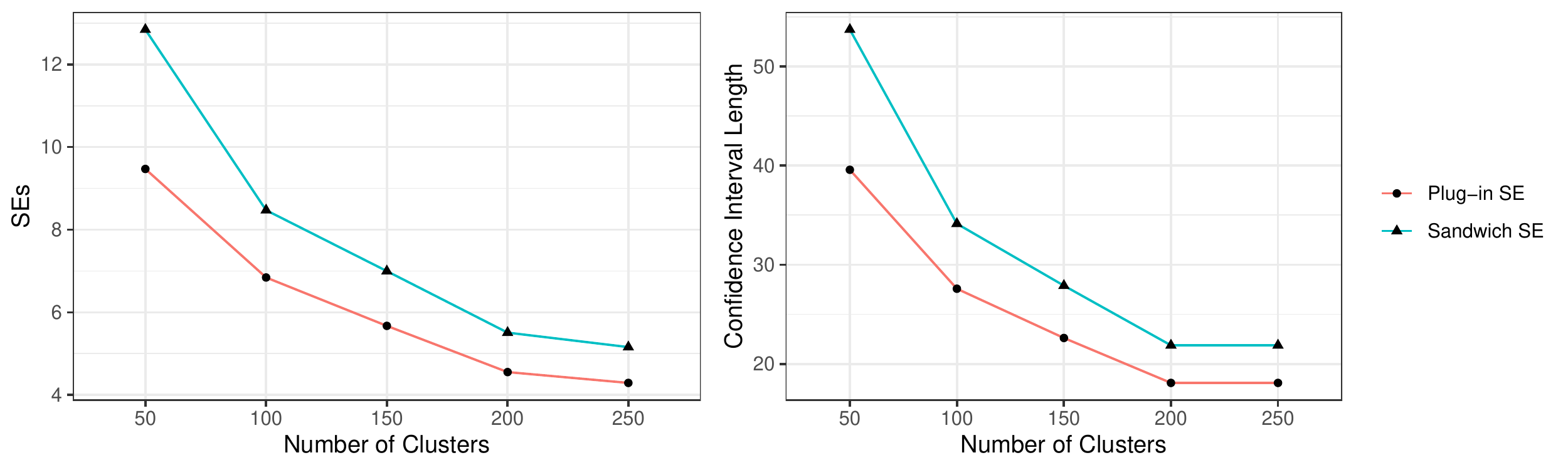}
    \caption{Comparative performance for two different variance estimators. -- medium overlap scenario.}
  \label{fig:se-med-overlap}
\end{figure}

\begin{figure}[htbp]
  \centering
    \includegraphics[scale=0.7]{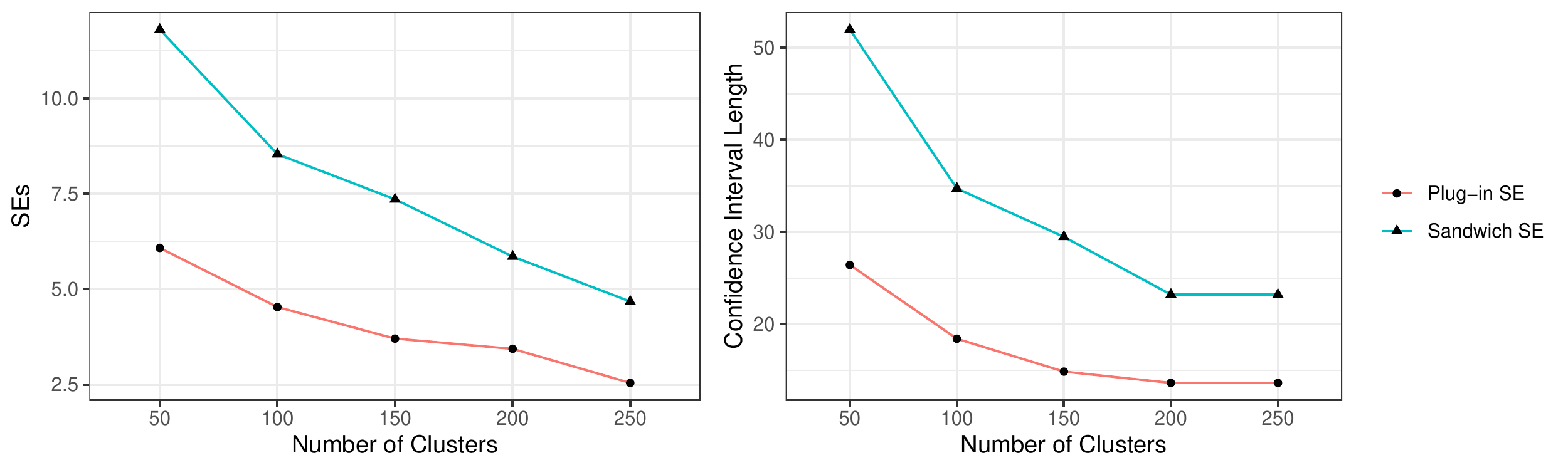}
    \caption{Comparative performance for two different variance estimators. -- poor overlap scenario.}
  \label{fig:se-poor-overlap}
\end{figure}

\end{document}